\documentclass[journal]{IEEEtran}
\usepackage{blindtext}
\usepackage{graphicx}
\usepackage[utf8]{inputenc}
\usepackage{graphics} 
\usepackage{epsfig} 
\usepackage{mathptmx} 
\usepackage{amsmath} 
\usepackage{amssymb}  
\usepackage{cite}
\usepackage{multicol}
\usepackage{mathtools, cuted}
\usepackage[cmintegrals]{newtxmath}
\usepackage{epstopdf}
\usepackage{graphics} 
\usepackage{epsfig} 

\usepackage{mathptmx} 
\usepackage{amsmath} 
\usepackage{amssymb}  
\usepackage{cite}
\usepackage{multicol}
\usepackage{mathtools, cuted}
\usepackage[cmintegrals]{newtxmath}
\usepackage{mathtools}

\usepackage{mathptmx}
\usepackage{helvet}
\usepackage{courier}
\usepackage{subfigure}
\usepackage{type1cm}
\usepackage{titlesec}
\usepackage{epsfig} 
\usepackage{mathptmx} 
\usepackage{amsmath} 
\usepackage{amssymb}  
\usepackage{cite}
\usepackage{multicol}
\usepackage{mathtools, cuted}
\usepackage[cmintegrals]{newtxmath}
\usepackage{makeidx}         
\usepackage{algorithm}
\usepackage{algpseudocode}

\usepackage{multicol}        
\usepackage[bottom]{footmisc}
\usepackage{caption}
\usepackage{nomencl}
\usepackage[usenames, dvipsnames]{color}

\usepackage{multicol}        
\usepackage[bottom]{footmisc}
\usepackage{caption}
\usepackage{titlesec}
\usepackage{nomencl}
\usepackage[usenames, dvipsnames]{color}
\usepackage{multirow}

\usepackage{amsthm}
 
\theoremstyle{definition}
\newtheorem{definition}{Definition}

\theoremstyle{theorem}
\newtheorem{theorem}{Theorem}

\theoremstyle{remark}

\theoremstyle{proposition}

\theoremstyle{corollary}

\theoremstyle{proof}

\newtheorem{assumption}{Assumption}
\theoremstyle{assumption}

\theoremstyle{lemma}

\ifCLASSINFOpdf
\else
\fi
\begin{document}
%
\title{Deep and Decentralized Multi-Agent Coverage of a Target with Unknown Distribution}
%
%
%

\author{Hossein Rastgoftar
\thanks{{\color{black}H. Rastgoftar is with the Department
of Aerospace and Mechanical Engineering, University of Arizona, Tucson,
AZ, 85721 USA e-mail: hrastgoftar@arizona.edu.}}
}
%
%

\markboth{
}%
{Shell \MakeLowercase{\textit{et al.}}: Bare Demo of IEEEtran.cls for IEEE Journals}
%



\maketitle
\begin{abstract}
This paper proposes a new architecture for multi-agent systems to cover an unknowingly distributed fast, safely, and decentralizedly. The inter-agent communication is organized by a directed graph with fixed topology, and we model agent coordination as a decentralized leader-follower problem with time-varying communication weights. Given this problem setting, we first present a method for converting communication graph into a neural network, where an agent can be represented by a unique node of the communication graph but multiple neurons of the corresponding neural network.  We then apply a mass-cetric strategy to  train time-varying communication weights of the neural network in a decentralized fashion which in turn implies that the observation zone of every follower agent is independently assigned by the follower based on positions of in-neighbors. By training the neural network, we can ensure safe and decentralized multi-agent  coordination  of coverage control.  Despite the target is unknown to the agent team,  we provide a proof for convergence of the proposed multi-agent coverage method. 
The functionality of the proposed method will be validated by a large-scale multi-copter team covering distributed targets on the ground.

\end{abstract}


\section{Introduction}\label{Introduction}
Multi-agent coberage has been received a lot of attentions by the control community over the recent years. 
Multi-agent coverage has many applications such as wildfire management \cite{9001230, 9147790}, border security \cite{pan2022mate}, agriculture \cite{din2022deep, davoodi2021graph}, and wildlife monitoring \cite{seraj2022multi}. A  variety of coverage approaches have been proposed by the researchers  that are reviewed in Section \ref{Related Work}.

\subsection{Related Work}\label{Related Work}
Sweep \cite{vasquez2022divide, bochkarev2016minimizing} and Spiral \cite{cabreira2018energy, choi2009online} are two available methods used for the single-vehicle coverage path planning, while  Vehicle Routing Problem \cite{toth2002overview, toth2002vehicle} is widely used for the  multi-agent coverage path planning. Diffusion-based multi-agent coverage convergence and stability are proposed in Ref. \cite{elamvazhuthi2018nonlinear}. Decentralized multi-agent coverage using local density feedback is  achieved by applying discrete-time mean-field model in Ref. \cite{biswal2021decentralized}.  Multi-agent coverage conducted by unicycle robots guided by a single leader is investigated in Ref. \cite{atincc2020swarm}, where the authors propose to decouple coordination and coverage modes. Adaptive decentralized multi-agent coverage is studied in  \cite{song2011decentralized, dirafzoon2011decentralized}. Ref. \cite{krishnan2022multiscale} offers a multiscale analysis of multi-agent coverage control that provides the convergence properties in continuous time. Human-centered active sensing of wildfire by unmanned aerial vehicles is studied in Ref. \cite{9147790}. Ref. \cite{9336858} suggests to apply k-means algorithm for planning of zone coverage by multiple agents. Reinforcement Learning- (RL-) based multi-agent coverage control is investigated by Refs. \cite{din2022deep, adepegba2016multi, dai2020graph, xiao2020distributed, kouzehgar2020multi}. Authors in  \cite{bai2021adaptive, nguyen2016discretized, abbasi2017new, luo2019voronoi, adepegba2016multi} used Vononoi-based approach for covering a distributed target.  Vononoi-based coverage  in the presence of obstacles and failures is presented as a leader-follower problem in Ref. \cite{bai2021adaptive}. Ref. \cite{patel2020multi} experimentally evaluate functionality Voronoi-based and other multi-agent coverage approaches in urban environment.

\subsection{Contributions}
This paper develops a method for decentralized multi-agent coverage  of a distributed target with an unknown distribution. We propose to define the inter-agent communications by a deep neural network, which is called \textit{coverage neural network},  with time-varying weights that are obtained such that coverage convergence is ensured.  
To this end, the paper establishes specific rules for structuring the coverage neural network  and proposes a mass-centric approach to train the network weights, at any time $t$, that specify inter-agent communication among the agent team. Although, the target is unknown to the agent team, we prove that the weights ultimately converge to the unique values that quantify target distribution in the motion space. The functionality of the proposed coverage method will be validated by simulating aerial coverage conducted by a team of  quadcopetr agents. 

Compared to the existing work, this paper offers the following novel contributions:
\begin{enumerate}
    \item The proposed multi-agent coverage approach learns the inter-agent communication weights in a forward manner as opposed to the existing neural learning problem, where they are trained by combining forward and backward iterations. More specifically, weights input to a hidden layer are assigned based on the (i) outputs of the previous layer and (i) target data information independently measured by observing the neighboring environment. We provide the proof of convergence for the proposed learning approach.
    \item The paper proposes a method for converting inter-agent communication graph into a neural network that will be used for organizing the agents, structuring the inter-agent communications, and partitioning the coverage domain.
    \item The paper develops a method for decentralized partitioning and coverage of an unknowingly distributed target. This method is indeed more computationally-efficient than the the available Voronoi-based partitioning methods that require all agents' positions to determine the search subdomain allocated to each individual agent.
\end{enumerate}

\subsection{Outline}
The remainder of the paper is organized as follows: The Problem Statement and Formulation are given in Section \ref{Problem Statement and Formulation}. The paper methodology is presented in Section \ref{Methodology}. Assuming every agent is a quadcopter, the multi-agent network dynamics is obtained in Section \ref{networkdynamics}, and followed by Simulation Results in Section \ref{Simulation} and Conclusion in Section \ref{Conclusion}.

\section{Problem Statement and Formulation}\label{Problem Statement and Formulation}
We consider a team of $N$ agents identified by set $\mathcal{V}=\left\{1,\cdots,N\right\}$ and classify them into the following three groups: 
\begin{enumerate}
\item ``boundary'' agents identified by $\mathcal{V}_B=\left\{1,\cdots,N_B\right\}$ are distributed along the boundary of the agent team configuration;
\item a single ``core'' agent identified by singleton $\mathcal{V}_C=\left\{N_B+1\right\}$ is an interior agent with the global position representing the global position of the agent configuration; and 
\item  follower agents defined by $\mathcal{V}_I=\left\{N_B+2,\cdots,N\right\}$ are all located inside the agent team configuration.
\end{enumerate}
Note that  $\mathcal{V}_B$, $\mathcal{V}_C$, and $\mathcal{V}_I$ are disjoint subsets of $\mathcal{V}$, i.e. $\mathcal{V}=\mathcal{V}_B\bigcup \mathcal{V}_C\bigcup \mathcal{V}_I$.
Inter-agent commucication among the agents are defined by graph $\mathcal{G}\left(\mathcal{V},\mathcal{E}\right)$ where $\mathcal{E}\subset \mathcal{V}\times \mathcal{V}$ defines edges of graph $\mathcal{G}$ and each edge represents a unique communication link (if $\left(j,i\right)\in \mathcal{E}$, then, $i$ accesses position of $j\in \mathcal{V}$). 
\begin{definition}\label{InNeighbor}
    We define
    \begin{equation}
        \mathcal{N}_i=\left\{j\in \mathcal{V}:\left(j,i\right)\in \mathcal{E}\right\},\qquad \forall i\in \mathcal{V}.
    \end{equation}
    as the set of in-neighbors of every agent $i\in \mathcal{V}$.
\end{definition}

\subsection{Neural Network Representation of Inter-Agent Communication} 
Graph $\mathcal{G}$ is defined such that it can be represented by a deep neural network with $M+1$ layers, where we use set  $\mathcal{M}=\left\{0,\cdots,M\right\}$ to define the layer identification numbers. Set $\mathcal{V}$ can be expressed as
\begin{equation}
    \mathcal{V}=\bigcup_{l\in \mathcal{M}}\mathcal{V}_l
\end{equation}
where $\mathcal{V}_0$ through $\mathcal{V}_M$ are disjoint subsets of $\mathcal{V}$. We use $\mathcal{W}_0$,  $\mathcal{W}_1$, $\cdots$, $\mathcal{W}_M$ to identify the neuron of layers $0$ through $M$ of the coverage neural network, and $\mathcal{W}_l$ and $\mathcal{V}_l$ are related by 
\begin{equation}
    \mathcal{W}_l=\begin{cases}
        \mathcal{V}_l&l\in \left\{0,M\right\}\\
        \mathcal{W}_{l-1}\bigcup\mathcal{V}_l&l\in \mathcal{M}\setminus\left\{0,M\right\}\\
    \end{cases}
    ,
\end{equation}
where $\mathcal{W}_0=\mathcal{V}_0=\mathcal{V}_B\bigcup\mathcal{V}_C$ defines neurons that uniquely represent boundary and core agents. 
\begin{definition}
For every neuron $i\in \mathcal{W}_l$ at layer $l\in \mathcal{M}\setminus \left\{0\right\}$, $\mathcal{I}_{i,l}\in \mathcal{W}_{l-1}$ defines those neurons of $\mathcal{W}_{l-1}$ that are connected to $i\in \mathcal{W}_l$.
Assuming the agent team forms an $n$-dimensional  configuration in a three-dimensional motion space ($n=2,3$),  we use the following key rules to define $\mathcal{I}_{i,l}$ for every $i\in \mathcal{W}_l$ and $l\in \mathcal{M}\setminus \left\{0\right\}$:
\begin{equation}\label{InNeighbor}
    \left|\mathcal{I}_{i,l}\right|=\begin{cases}
    1&\mathrm{If}~i\in \mathcal{W}_{l-1}\bigcap  \mathcal{W}_{l}\mathrm{~and~}l\in \mathcal{M}\setminus \left\{0\right\}\\
     n+1&\mathrm{If}~i\in \mathcal{W}_{l}-\mathcal{W}_{l-1}\mathrm{~and~}l\in \mathcal{M}\setminus \left\{0\right\}\\
     n+1&\mathrm{If}~i\in \mathcal{W}_{M}\\
  0&\mathrm{If}~i\in \mathcal{W}_0\\
    \end{cases}
    .
\end{equation}
\end{definition}

\begin{figure}[h]
\centering
\subfigure[]{\includegraphics[width=0.32\linewidth]{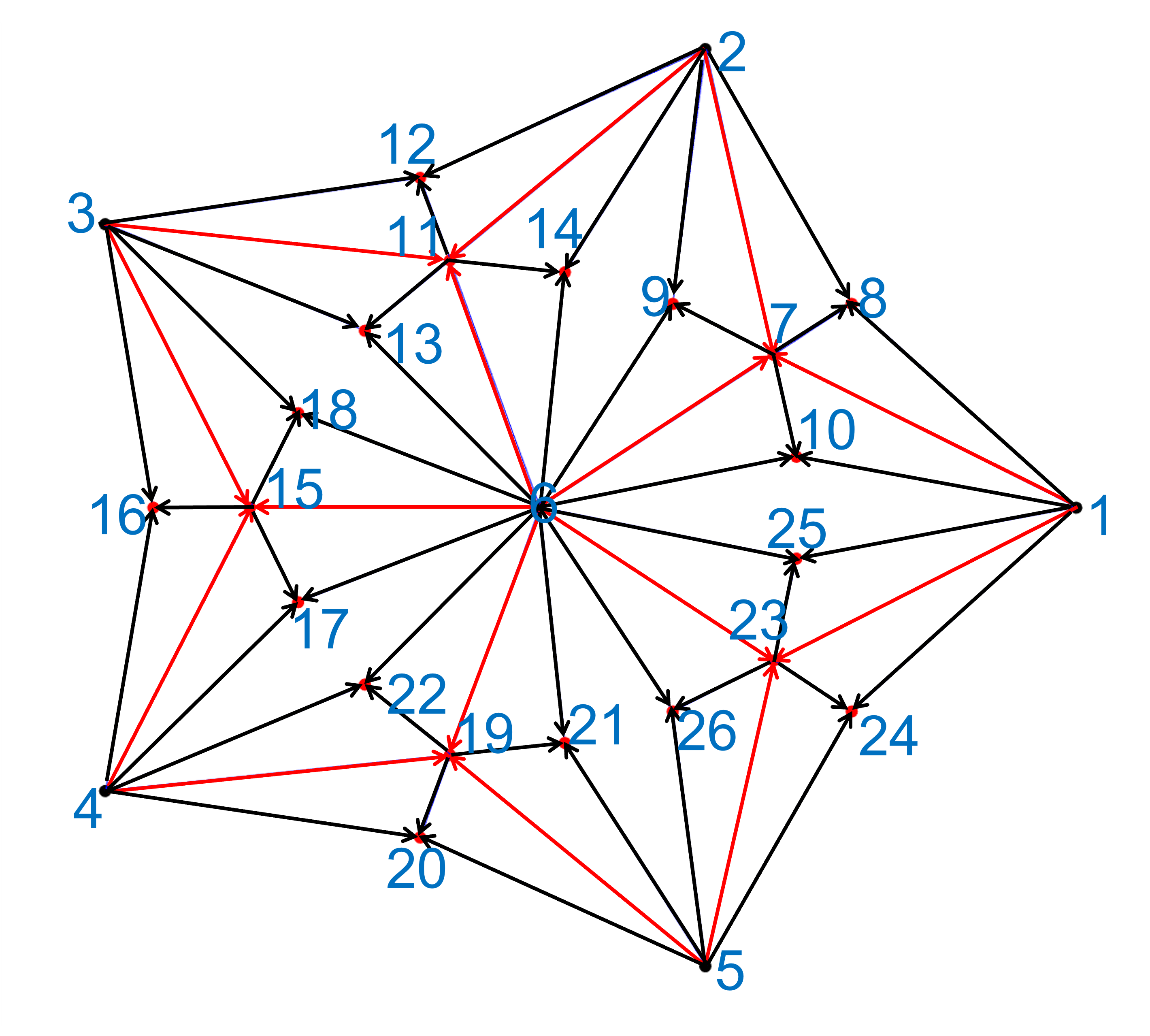}}
 \subfigure[]{\includegraphics[width=0.66\linewidth]{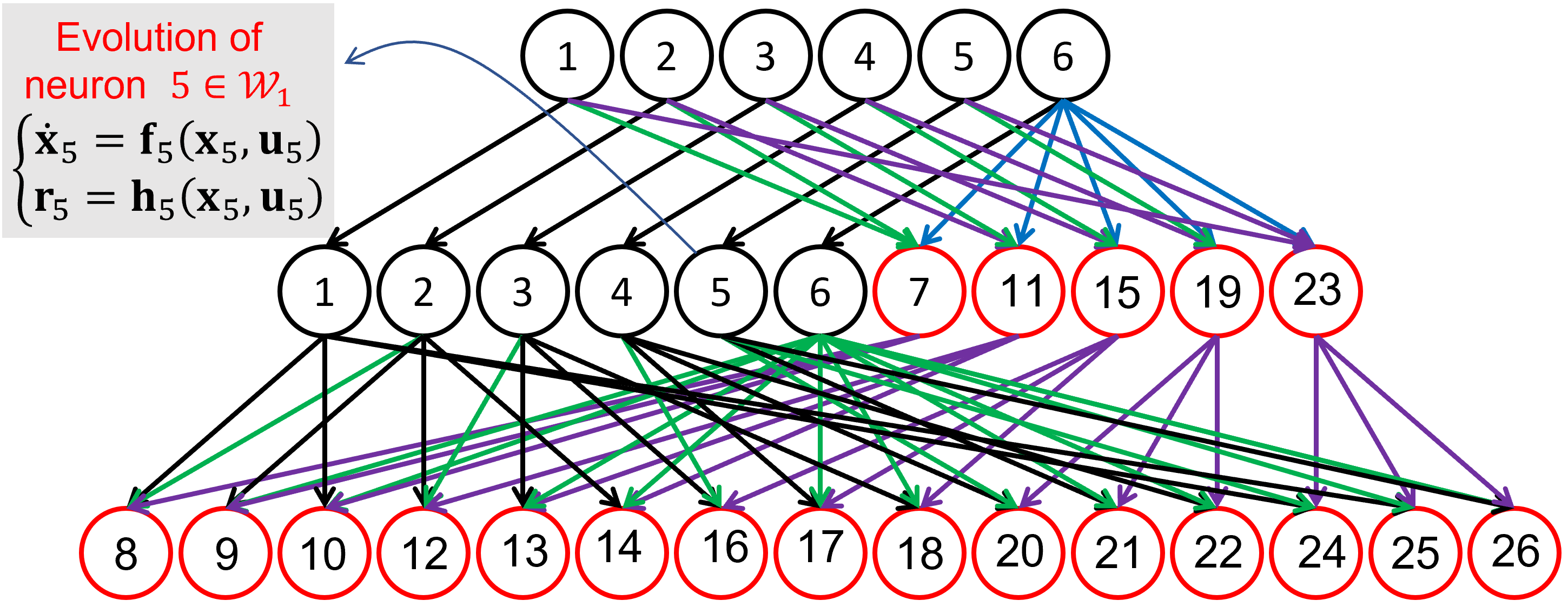}}
\caption{(a) Graph $\mathcal{G}\left(\mathcal{V},\mathcal{E}\right)$ representing the coverage neural network that consists g of three layers ($\mathcal{M}=\{0,1,2\}$), where $\mathcal{V}=\mathcal{W}_0\bigcup\mathcal{W}_1\bigcup\mathcal{W}_2$. (b) The agent team configuration suggested for a $2$-D coverage.}
\label{Spatial-Temporal-0}
\end{figure}

We note that $\mathcal{N}_i$ and $\mathcal{I}_{i,l}$ can be related by
\begin{equation}
    \bigwedge_{l\in \mathcal{M}\setminus \left\{0\right\}}\bigwedge_{i\in \mathcal{W}_l- \mathcal{W}_{l-1}}\left(\mathcal{I}_{i,l}=\mathcal{N}_i\right).
\end{equation}

For better clarification, we consider an agent team with $N=26$ agents identified by set $\mathcal{V}=\left\{1,\cdots,26\right\}$ forming a two-dimensional configuration ($n=2$)  shown in Fig. \ref{Spatial-Temporal-0} (a). The inter-agent communications shown in Fig. \ref{Spatial-Temporal-0} (a)  can be represented by the neural network of Fig. \ref{Spatial-Temporal-0} (b) with three layers $\mathcal{M}=\left\{0,1,2\right\}$, where $\mathcal{W}_0=\left\{1,\cdots,6\right\}$, defining the boundary and core leaders, has no in-neighbors  $\mathcal{W}_2=\left\{8,9,10, 12,13,14,16,17,18,20,21,22, 24,25,26\right\}$ defining followers, each has three in-neighbors. Also, $\left\{7,11,15,19,23\right\}\in \mathcal{W}_1$ each has three in-neighbors but the remaining neurons of $\left\{1,\cdots,6\right\}$,  that are repeated in layer $0$, each has one in-neighbor.

\begin{figure}[ht]
\center
\includegraphics[width=3.3 in]{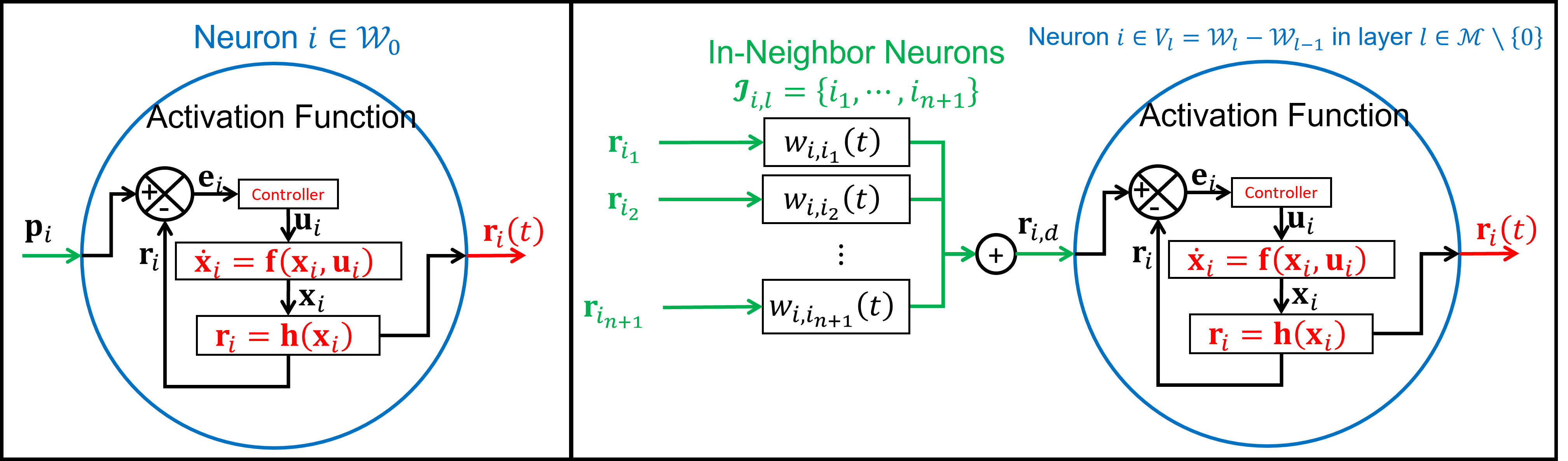}
\caption{Activation functions, inputs and outputs of neurons of the coverage neural network.  (Left): An example neuron in $i\in\mathcal{W}_0$. (Right): An example neuron $i\in \mathcal{W}_l-\mathcal{W}_{l-1}$ in layer $l\in \mathcal{M}\setminus \left\{0\right\}$ }
\label{milestoes}
\end{figure}

\subsection{Differential Activation Function} Unlike the available neural network, the activation of the coverage network's neurons are operated differential activation functions given by  nonlinear dynamics
\begin{equation}\label{rid}
    \begin{cases}
        \dot{\mathbf{x}}_i=\mathbf{f}_i\left(\mathbf{x}_i,\mathbf{u}_i\right)\\
        \mathbf{r}_i=\mathbf{h}_i\left(\mathbf{x}_i\right)
    \end{cases}
    ,\qquad i\in \mathcal{W}_l,~l\in \mathcal{M},
\end{equation}
that  is used to model the agent $i\in \mathcal{V}_l$ (See Fig. \ref{milestoes}), where  $\mathbf{x}_i\in \mathbb{R}^{n_{x,i}}$  and $\mathbf{u}_i\in \mathbb{R}^{n_{u,i}}$ denote the state vector and the control of neuron $i$, respectively, and $\mathbf{h}_i:\mathbb{R}^{n_{x,i}}\rightarrow  \mathbb{R}^3$, $\mathbf{f}_i:\mathbb{R}^{n_{x,i}}\rightarrow  \mathbb{R}^{n_{x,i}}$, and $\mathbf{g}_i:\mathbb{R}^{n_{x,i}}\rightarrow  \mathbb{R}^{n_{x,i}\times n_{u,i}}$ are smooth functions.

\textit{The output of neuron $i$} denoted by $\mathbf{r}_i\in \mathbb{R}^{3\times 1}$ is the position of agent $i$. \textit{The input of neuron $i$} is defined by 
\begin{equation}\label{riddd}
    \mathbf{r}_{i,d}(t) = \begin{cases} 
    \mathbf{p}_i\quad \textrm{(given)} & i \in \mathcal{W}_0 \\
    \sum_{j \in \mathcal{I}_{i,l}} w_{ij}(t)\mathbf{r}_j(t) & i \in \mathcal{W}_l-\mathcal{W}_{l-1},~l\in \mathcal{M}\setminus \left\{0\right\}
    \end{cases}
\end{equation}
where $\mathbf{p}_i$ is a desired  constant position for leader agent $i \in \mathcal{W}_0$. 
Also, $w_{i,j}(t) > 0$ is the time-varying communication weight between $i\in \mathcal{W}_l$ and $j \in \mathcal{I}_{i,l}$, and satisfies the following constraint:
\begin{equation}\label{weq}
    \bigwedge_{l\in \mathcal{M}\setminus \left\{0\right\}}\bigwedge_{i\in \mathcal{W}_l- \mathcal{W}_{l-1}}\left(\sum_{j\in \mathcal{I}_{i,l}}w_{i,j}(t)=1\right),\qquad \forall t.
\end{equation}



\subsection{Objectives} Given above problem setting, this paper offers a neural-network-based method  for optimal coverage of  target set $\mathcal{D}$  with unknown distribution in a $3$-dimensional motion space.  To achieve this objective, we assume that positions of boundary leader agents, defined by $\mathcal{W}_0\setminus \mathcal{V}_C$, are known, and solve the following two main problems:
\begin{enumerate}
    \item \textbf{Problem 1--Abstract Representation of Target:} We develop a mass-centric approach in Section \ref{Assignment of Leaders' Positions} to abstractly represent target by $N-N_B+1$ position vectors $\mathbf{p}_{N_B+2}$ through $\mathbf{p}_N$ that are considered as followers' desired  positions.
    \item \textbf{Problem 2--Decentralized Target Acquisition:} We propose a forward method to train the communication weights $w_{i,j}(t)$, and assign control input $\mathbf{u}_i$, for every agent $i\in \mathcal{V}$ and in-neighbor agent $j\in \mathcal{I}_{i,l}$,  such that actual position $\mathbf{r}_i$ converges to the desired position $\mathbf{p}_i$ in a decentralized fashion, for every $i\in \mathcal{V}\setminus \mathcal{W}_0$, where $i\in \mathcal{V}\setminus \mathcal{W}_0$ does \underline{\textbf{not}} know global position $\mathbf{p}_j(t)$ of any in-neighbor agent $j\in \mathcal{V}$.
\end{enumerate}


 Without loss of generality, $n$ is either $2$, or $3$ because motion space is three-dimensional. More specifically, for ground coverage $n=2$ and $\mathcal{D}$
 specifies finite number of targets on the ground.
\section{Methodology}\label{Methodology}

The agent team is aimed to cover a zone that is specified by $\mathcal{D}=\left\{1,\cdots,n_d\right\}$,  where $\mathbf{d}_i\in \mathbb{R}^{3\times1}$ is the position of target $i\in \mathcal{D}$. 
We also define \textit{intensity} function $\mathcal{T}:\mathcal{D}\rightarrow \left(0,1\right]$ to quantify the intensity of data point $i\in \mathcal{D}$ positioned at $i\in \mathcal{D}$.

For development of the neural-networ-based coverage model, we apply the following Definitions and Assumptions:
\begin{assumption}\label{leadersrankcondition}
    Boundary leader agents form an $n-D$ polytope in $\mathbb{R}^{n}$, thus, the boundary agents' desired positions 
    must satisfy the following rank condition:
    \begin{equation}
        \mathrm{rank}\left(\begin{bmatrix}
            \mathbf{p}_2-\mathbf{p}_1&\cdots&\mathbf{p}_{N_B}-\mathbf{p}_1
        \end{bmatrix}\right)
        =n
    \end{equation}    
\end{assumption}
The polytope defined by the boundary agents is called \textit{leading polytope}. 
\begin{assumption}\label{assumsimplex}
The leading polytope, defined by the boundary agents, can be decomposed into $N_L$ disjoint $n$-dimensional simplexes all sharing the core node $N_B+1\in \mathcal{W}_0$. 
\end{assumption}
We let $\mathcal{L}=\left\{1,\cdots,N_L\right\}$ define all simplex cells of the leading polytope, where $\mathcal{S}_i=\left\{h_{i,1},\cdots,h_{i,n},N_B+1\right\}$ defines vertices of simplex cell $i\in \mathcal{L}$, i.e. $h_{i,1},\cdots,h_{n,i}\in \mathcal{S}_i\setminus \left\{N_B+1\right\}\subset \mathcal{W}_0$ are the boundary nodes of simplex $i\in \mathcal{L}$. Per Assumption \ref{assumsimplex}, we can write
\begin{subequations}
    \begin{equation}
        \mathcal{W}_0=\bigcup_{i\in \mathcal{L}}\mathcal{S}_i,
    \end{equation}
     \begin{equation}
        \bigwedge_{i\in \mathcal{L}}\left(\mathrm{rank}\left(\begin{bmatrix}
            \mathbf{p}_{h_{i,1}}-\mathbf{p}_{N_B+1}&\cdots& \mathbf{p}_{h_{i,n}}-\mathbf{p}_{N_B+1}
        \end{bmatrix}\right)
        =n\right).
    \end{equation}    
\end{subequations}

\begin{algorithm}
  \caption{Assignment of followers' ``desired'' communication weights and ``desired'' positions.}\label{alg2}
  \begin{algorithmic}[1]
          \State \textit{Get:} Dataset $\mathcal{D}$; set $\mathcal{W}_0$, $\cdots$, $\mathcal{W}_M$; and graph $\mathcal{G}\left(\mathcal{V},\mathcal{E}\right)$.
         \State \textit{Obtain:} Desired position of every agent $i\in \mathcal{V}\setminus \mathcal{W}_0$ and desired communication weights.
          \For{\texttt{ $l\in \mathcal{M}\setminus \left\{0\right\}$}}
            \For{\texttt{ $i\in \mathcal{W}_l-\mathcal{W}_{l-1}$}}
                \State Obtain $\bar{\mathcal{D}}_i$ by using Eq. \eqref{bardi}.        
                \State Choose $\mathbf{p}_i$ as the centroid of $\bar{\mathcal{D}}_i$.
                \State Get desired positions of in-neighbor agents $\mathcal{I}_{i,l}$.   
                \State Compute desired communication weights by \eqref{desiredcomweights}.   
           \EndFor
        \EndFor    
  \end{algorithmic}
\end{algorithm}

\begin{assumption}\label{In-neighbor-three}
Every agent $i\in \mathcal{V}\setminus \mathcal{W}_0$ has $n+1$ in-neighbors, therefore,
\begin{equation}
    \bigwedge_{l\in \mathcal{M}\setminus\left\{0\right\}}\bigwedge_{i\in \mathcal{W}_l-\mathcal{W}_{l-1}}\left(\left|\mathcal{I}_{i,l}\right|=n+1\right).
\end{equation}
\end{assumption}
\begin{assumption}\label{nonsingularity}
The in-neighbors of every agent $i\in \mathcal{V}\setminus \mathcal{W}_0$ defined by $\mathcal{N}_i=\left\{j_1,\cdots,j_{n+1}\right\}$ forms an $n$-D simplex. This condition can be formally specified as follows:
 \begin{equation}
        \bigwedge_{l\in \mathcal{M}\setminus\left\{0\right\}}\bigwedge_{i\in \mathcal{W}_l-\mathcal{W}_{l-1}}\left(\mathrm{rank}\left(\begin{bmatrix}
            \mathbf{p}_{j_2}-\mathbf{p}_{j_1}&\cdots&\mathbf{p}_{{j_{n+1}}}-\mathbf{p}_{j_1}
        \end{bmatrix}\right)
        =n\right).
    \end{equation}    

\end{assumption}
\begin{definition}\label{deconvex}
For every agent $i\in \mathcal{V}\setminus \mathcal{W}_0$,
\begin{subequations}
\begin{equation}\label{barci}
 \resizebox{0.99\hsize}{!}{%
$
    \bar{\mathcal{C}}_i=\left\{\sum_{j\in \mathcal{I}_{i,l}}\sigma_j\mathbf{p}_j:\sigma_j\geq0~\mathrm{and}~\sum_{j\in \mathcal{I}_{i,l}}\sigma_j=1\right\},~i\in \mathcal{W}_l-\mathcal{W}_{l-1},~l\in \mathcal{M},
    $}
\end{equation}
\begin{equation}
 \resizebox{0.99\hsize}{!}{%
$
    {\mathcal{C}}_i(t)=\left\{\sum_{j\in \mathcal{I}_{i,l}}\sigma_j\mathbf{r}_j(t):\sigma_j\geq0~\mathrm{and}~\sum_{j\in \mathcal{I}_{i,l}}\sigma_j=1\right\},~i\in \mathcal{W}_l-\mathcal{W}_{l-1},~l\in \mathcal{M},
    $}
\end{equation}
\end{subequations}
define the convex hulls specified by ``desired'' and ``actual'' positions of agent $i$'s in-neighbors, respectively.
\end{definition}
\begin{definition}
    We define 
    \begin{subequations}        
    \begin{equation}
        \mathcal{C}=\bigcup_{l\in \mathcal{M}\setminus \left\{0\right\}}\bigcup_{i\in \mathcal{W}_l-\mathcal{W}_{l-1}}\bar{\mathcal{C}}_i
    \end{equation}
    \begin{equation}
        \mathcal{C}=\bigcup_{l\in \mathcal{M}\setminus \left\{0\right\}}\bigcup_{i\in \mathcal{W}_l-\mathcal{W}_{l-1}}\mathcal{C}_i(t)
    \end{equation}
    \end{subequations}
    specify the coverage zone that enclose all data points defined by set $\mathcal{D}$.
\end{definition}

By considering Definition \ref{deconvex}, we can express set $\mathcal{D}$ as
\[
        {\mathcal{D}}=\bigcup_{i\in \mathcal{I}_{i,l}}\bar{\mathcal{D}}_i\qquad \mathrm{or}\qquad \mathcal{D}=\bigcup_{i\in \mathcal{I}_{i,l}}{\mathcal{D}}_i(t),
   \]
   where 
   \begin{equation}\label{bardi}
    \bar{\mathcal{D}}_i=\left\{j\in \mathcal{D}:\mathbf{d}_j\in \bar{\mathcal{C}}_i\right\},
\end{equation}
   is the target set that is ``desired'' to be searched by follower agent $i\in \mathcal{V}\setminus \mathcal{V}_0$ whereas  
   \begin{equation}\label{di}
   {\mathcal{D}}_i(t)=\left\{{j\in \mathcal{D}:\mathbf{d}_j\in \mathcal{C}}_i(t)\right\},
\end{equation} 
is the subset of $\mathcal{D}$ that is ``actually'' searched by follower agent $i\in \mathcal{V}\setminus \mathcal{V}_0$ at time $t$. Note that $\bar{\mathcal{D}}_i$ and ${\mathcal{D}}_i(t)$ are enclosed by the  convex hulls $\bar{\mathcal{C}}_i$ and ${\mathcal{C}}_i(t)$, respectively, that are determined by the ``desired'' and ``actual'' positions of the agent $i\in \mathcal{V}\setminus \mathcal{W}_0$, respectively.

\begin{assumption}\label{keyassumption}
    We assume that $\bar{\mathcal{D}}_i\neq \emptyset$ and $\mathcal{D}_i(t)\neq \emptyset$, at any time $t$, for every $i\in \mathcal{V}\setminus \mathcal{W}_0$.
\end{assumption}
In order to assure that Assumption \ref{keyassumption} is satisfied, we may need to regenerate target set $\mathcal{D}$, when target data set $\mathcal{D}$ is scarcely distributed. When this regeneration is needed, we first convert discrete set $\mathcal{D}$ to discrete set
\begin{equation}
    \mathcal{D}'=\left\{\mathbf{d}=\sum_{i=1}^{n_d}\mathcal{N}\left(\mathbf{r}; \mathbf{d}_i,\mathbf{\Sigma}_i\right):\mathbf{d}_i\in \mathcal{D},~\mathbf{r}\in \mathcal{C}\right\}
\end{equation}
where $\mathcal{N}\left(\mathbf{r}; \mathbf{d}_i,\mathbf{\Sigma}_i\right)$ is a multi-variate normal distribution specified by mean vector $\mathbf{d}_i$ and covariance matrix $\mathbf{\Sigma}_i$. Then, we regenerate $\mathcal{D}$ by uniform dicretization of $\mathcal{D}$.




\subsection{Abstract Representation of Target Locations}\label{Assignment of Leaders' Positions}
 We use the approach presented in Algorithm \ref{alg2} to abstractly represent target set $\mathcal{D}$ by position vectors $\mathbf{p}_{N_B+2}$, $\cdots$, $\mathbf{p}_N$, given (i)
 desired positions of leader agents denoted $\mathbf{p}_1$ through $\mathbf{p}_{N_B+1}$, (ii) the edge set $\mathcal{E}$, and (iii) target set $\mathcal{D}$, as the input.  Note that $\mathbf{p}_i$ is considered the global desired position of follower  $i\in \mathcal{V}_I=\left\{N_B+2,\cdots,N\right\}$, but no follower $i\in \mathcal{V}\setminus \mathcal{V}_0$ knows $\mathbf{p}_i$.

The desired position of every follower agent $i\in \mathcal{V}_I=\mathcal{V}\setminus \mathcal{W}_0$ is obtained by
\begin{equation}\label{followerdesired position}
    {\mathbf{p}}_{i}=
        {\sum_{h\in \bar{\mathcal{D}}_i}{\mathcal{T}_h\left(h\right)\mathbf{d}}_h\over \left|\bar{\mathcal{D}}_i\right|},\qquad \forall i\in \mathcal{V}\setminus \mathcal{W}_0,
\end{equation}
where $\bar{\mathcal{D}}_i$, defined by Eq. \eqref{bardi}, is a target data subset that is  enclosed by $\bar{\mathcal{C}}_i$ and defined by Eq. \eqref{barci}. We notice that the desired position of every follower agent $i\in \mathcal{V}\setminus \mathcal{W}_0$ is assigned in a ``forward'' manner which in turn implies that $\mathcal{W}_l$'s desired positions are assigned after determining  $\mathcal{W}_{l-1}$'s desired positions, for every $l\in \mathcal{M}\setminus \left\{0\right\}$.

Given desired positions of every follower agent $i\in \mathcal{V}\setminus \mathcal{W}_0$ and every in-neighbor agent $j\in \mathcal{N}_i$, $\varpi_{i,j}>0$ defines \textbf{the desired communication weight} between $i\in \mathcal{V}\setminus \mathcal{W}_0$ and $j\in \mathcal{N}_i$, and is obtained by solving $n+1$ linear algebraic equations provided by
\begin{subequations}\label{desiredcomweights}
    \begin{equation}
        \mathbf{p}_i=\sum_{j\in \mathcal{I}_{i,l}}\varpi_{i,j}\mathbf{p}_j,
    \end{equation}
        \begin{equation}
        \sum_{j\in \mathcal{I}_{i,l}}\varpi_{i,j}=1.
    \end{equation}
\end{subequations}
Algorithm \ref{alg2} also presents our proposed hierarchical approach for assignment of followers' desired communication weights.
\begin{definition}
    We define desired weight matrix $\bar{\mathbf{L}}=\left[\bar{L}_{ij}\right]\in \mathbb{R}^{N\times N}$ with $(i,j)$ entry
    \begin{equation}
        \bar{L}_{ij}=\begin{cases}
            \varpi_{i,j}&i\in \mathcal{V}\setminus \mathcal{W}_0,~j\in \mathcal{N}_i\\
            -1&i=j\\
            0&\mathrm{otherwise}
        \end{cases}
        .
    \end{equation}
\end{definition}

\subsection{Decentralized Target Acquisition}\label{Safe Decentralized Target Acquisition}
For a decentralized coverage, it is necessary that every follower agent $i\in \mathcal{V}_l=\mathcal{W}_l-\mathcal{W}_{l-1}$, represented by a neorn in layer $l\in \mathcal{M}\setminus \left\{0\right\}$, chooses control $\mathbf{u}_i\in \mathbb{R}^{n_u\times 1}$, based on actual positions of the in-neighbor agents $\mathcal{I}_{i,l}$, such that $\mathbf{r}_i(t)$ stably tracks $\mathbf{r}_{i,d}(t)$ that is defined by Eq. \eqref{riddd}.  Note that $\mathbf{r}_{i,d}(t)$ is a linear combination of the in-neighbors' actual positions, for $i\in \mathcal{V}\setminus \mathcal{W}_0$, with  (communication) weights that are time-varying and constrained to satisfy equality constraint \eqref{weq}. 
We use forward training to learn the coverage neural network. This means that communication weights of layer $l \in \mathcal{M}\setminus \left\{ 0 \right\}$ neurons are assigned before communication weights of layer $l+1 \in \mathcal{M}\setminus \left\{ 0,M\right\}$ neurons, where communication weight of neuron $i \in \mathcal{V}_l=\mathcal{W}_l-\mathcal{W}_{l-1}$ is learned by solving a quadratic program. Let
\begin{equation}\label{followerdesired position}
    \bar{\mathbf{r}}_{i}(t)=
        {\sum_{h\in \mathcal{D}_i(t)}{\mathcal{T}(h)\mathbf{d}}_h(t)\over \left|\mathcal{D}_i(t)\right|},\qquad i\in \mathcal{V}_l,~l\in \mathcal{M}\setminus\left\{0\right\},
\end{equation}
denote the cetroid of subset set $\mathcal{D}_i(t)\subset \mathcal{D}$, where  $\mathcal{D}_i(t)\subset \mathcal{D}$ is defined (obtained) by Eq. \eqref{di}. Then, followers' communication weights are determined by minimizing
\begin{equation}\label{QP}
    \min \|\sum_{h\in \mathcal{I}_{i,l}}w_{i,h}(t)\mathbf{r}_j-\bar{\mathbf{d}}_{i}(t\|^2
\end{equation}
subject to equality constraint \eqref{weq}.

\begin{definition}
    We define  weight matrix ${\mathbf{L}}=\left[{L}_{ij}\right]\in \mathbb{R}^{N\times N}$ with $(i,j)$ entry
    \begin{equation}
        {L}_{ij}=\begin{cases}
            w_{i,j}&i\in \mathcal{V}\setminus \mathcal{W}_0,~j\in \mathcal{N}_i\\
            -1&i=j\\
            0&\mathrm{otherwise}
        \end{cases}
        .
    \end{equation}
\end{definition}

\begin{theorem}
Assume every agent $i\in \mathcal{V}$ chooses control input $\mathbf{u}_i$ such that $\mathbf{r}_i(t)$ asymptotically tracks $\mathbf{r}_{i,d}(t)$. Then, $\mathbf{r}_i(t)$ asymptotically converges to the desired position $\mathbf{p}_i$ for every $i\in \mathcal{V}$.
\end{theorem}
\begin{proof}
If every agent $j\in \mathcal{W}_0$ asymptotically tracks $\mathbf{r}_{j,d}(t)$, then, actual position $\mathbf{r}_j$ converges to $\mathbf{p}_j$ because $\mathbf{r}_{j,d}(t)=\mathbf{p}_j$ is constant per Eq. \eqref{riddd}. Then, for every $i\in \mathcal{W}_1$, vertices of the simplex $\bar{\mathcal{C}}_i$, belonging to $\mathcal{W}_0$, asymptotically converge to the vertices $\bar{\mathcal{C}}_i$, where $\bar{\mathcal{C}}_i$ and $\mathcal{C}_i$ enclose target data subsets $\bar{\mathcal{D}}_i$ and $\mathcal{D}_i$, respectively. This implies that $\mathbf{r}_{i,d}(t)$, defined as the centroid of $\mathcal{D}_i(t)$ asymptotically converges to $\mathbf{p}_i$ for every $i\in \mathcal{W}_1$. By extending this logic, we can say that this convergence is propagated through the feedforward network $\mathcal{G}\left(\mathcal{V},\mathcal{E}\right)$.  As the result, for every agent $i\in \mathcal{W}_l$ and layer $l\in \mathcal{M}\setminus \left\{0\right\}$, vertices of simplex $\mathcal{C}_i(t)$ asymptotically converge the vertices of   $\bar{\mathcal{C}}_i$ which in turn implies that $\mathbf{r}_{i,d}(t)$ asymptotically converges to $\mathbf{p}_i$. This also implies that $\mathbf{r}_i$ asymptotically converges to $\mathbf{p}_i$ per the theorem's assumption.

\end{proof}

\section{Network Dynamics}\label{networkdynamics}
In this section, we suppose that every agent is a quacopter and 
use the input-state feedback linearization presented in  \cite{rastgoftar2021safe, asslouj2022quadcopter} and summerized in the Appendix to  model quadcopter motion by the fourth-order dynamics \eqref{quaddynamicsext} in the Appendix. Here, we propose  to choose $\mathbf{v}_i$ as follows:
\begin{equation}
    \mathbf{v}_i=-k_{1,i}\dddot{\mathbf{r}}_i-k_{2,i}\ddot{\mathbf{r}}_i-k_{3,i}\dot{\mathbf{r}}_i+k_{4,i}\left({\mathbf{r}}_{i,d}(t)-\mathbf{r}_i\right), \qquad i\in \mathcal{V},
\end{equation}
where $\mathbf{r}_{i,d}(t)$ is defined by Eq. \eqref{riddd}. Then, the external dynamics of the quadcopter team is given by \cite{rastgoftar2021safe, rastgoftar2022integration}
\begin{equation}
\label{mainlinearized}
\dfrac{d}{dt}\left(
\begin{bmatrix}
\mathbf{Y}\\
\dot{\mathbf{Y}}\\
\ddot{\mathbf{Y}}\\
\dddot{\mathbf{Y}}\\
\end{bmatrix}
\right)
=\mathbf{A}_{\mathrm{MQS}}
\begin{bmatrix}
\mathbf{Y}\\
\dot{\mathbf{Y}}\\
\ddot{\mathbf{Y}}\\
\dddot{\mathbf{Y}}\\
\end{bmatrix}
+
\mathbf{B}_{\mathrm{MQS}}
\begin{bmatrix}
\mathbf{R}_L\\
\dot{\mathbf{R}}_L\\
\ddot{\mathbf{R}}_L\\
\dddot{\mathbf{R}}_L\\
\end{bmatrix}
,
\end{equation}
where $\mathbf{Y}=\mathrm{vec}\left(\begin{bmatrix}
\mathbf{r}_1&\cdots&\mathbf{r}_N
\end{bmatrix}^T\right)$,  $\mathbf{R}_L=\mathrm{vec}\left(\begin{bmatrix}
\mathbf{p}_1&\cdots&\mathbf{p}_{N_B+1}
\end{bmatrix}^T\right)$, $\mathbf{L}_0=\begin{bmatrix}\mathbf{I}_{N_B+1}&\mathbf{0}_{\left(N_B+1\right)\times \left(N-N_B-1\right)}\end{bmatrix}^T\in \mathbb{R}^{N\times\left(N_B+1\right)}$,
\begin{subequations}
\begin{equation}
\mathbf{A}_{\mathrm{MQS}}=
    \begin{bmatrix}
\mathbf{0}&\mathbf{I}_{3N}&\mathbf{0}&\mathbf{0}\\
\mathbf{0}&\mathbf{0}&\mathbf{I}_{3N}&\mathbf{0}\\
\mathbf{0}&\mathbf{0}&\mathbf{0}&\mathbf{I}_{3N}\\
\mathbf{I}_3\otimes \left( \mathbf{K}_4 \mathbf{L}\right)&-\mathbf{K}_3\mathbf{I}_{3N}&-\mathbf{K}_2\mathbf{I}_{3N}&-\mathbf{K}_1\mathbf{I}_{3N}
\end{bmatrix}
,
\end{equation}
\begin{equation}
\resizebox{0.99\hsize}{!}{%
$
\mathbf{B}_{\mathrm{MQS}}=
    \begin{bmatrix}
\mathbf{0}&\mathbf{0}&\mathbf{0}&\mathbf{0}\\
\mathbf{0}&\mathbf{0}&\mathbf{0}&\mathbf{0}\\
\mathbf{0}&\mathbf{0}&\mathbf{0}&\mathbf{0}\\
\mathbf{I}_3\otimes\left( \mathbf{K}_{4} \mathbf{L}_0\right)&\mathbf{I}_3\otimes \left(\mathbf{K}_{3} \mathbf{L}_0\right)&\mathbf{I}_3\otimes\left( \mathbf{K}_2\mathbf{L}_0\right)&\mathbf{I}_3\otimes \left( \mathbf{K}_1 \mathbf{L}_0\right)
\end{bmatrix}
,
$
}
\end{equation}
\begin{equation}
    j=1,2,3,4,\qquad \mathbf{K}_j=\mathrm{diag}\left(k_{j,1},\cdots,k_{j,N}\right),
\end{equation}
\end{subequations}
$\mathbf{I}_{3N}\in \mathbb{R}^{3N\times 3N}$ is the identity matrix, and ``vec'' is the matrix vectorization operator.
Note that control gains $k_{j,i}$ ($i\in \mathcal{V}$ and $j=1,2,3,4$) are selected such that roots of the characteristic equation
\begin{equation}
\label{CharEq}
    \left|s^4\mathbf{I}+s^3\mathbf{K}_1\mathbf{L}+s^2\mathbf{K}_2\mathbf{L}+s\mathbf{K}_3+\mathbf{K}_4\right|=0
\end{equation}
are all located the collective dynamics \eqref{mainlinearized} is stable.  

\section{Simulation Results}\label{Simulation}
We consider an agent team consisting of $57$ quadcopters with the reference configuration shown in Fig. \ref{ReferenceConfiguration}, where we use the model and trajectory control presented in Refs. \cite{rastgoftar2021safe, asslouj2022quadcopter} for multi-agent coverage simulation. Here quadcopters $1$ through $4$ defined by set $\mathcal{V}_B=\left\{1,2,3,4\right\}$ are the boundary leader agents; agent $5$ defined by singleton $\mathcal{V}_C=\left\{5\right\}$ is core leader; and the remaining agents defined by $\mathcal{V}_I=\left\{6,\cdots,57\right\}$ are followers.

The inter-agent communications are directional and shown by blue vectors in Fig. \ref{ReferenceConfiguration}. The communication graph is defined by $\mathcal{G}\left(\mathcal{V},\mathcal{E}\right)$ and converted into the  neural network shown in Fig. \ref{NNResults} with four layers, thus,  $\mathcal{M}=\left\{0,1,2,3\right\}$ ($M=3$), and $\mathcal{V}$ can be expressed as $\mathcal{V}=\mathcal{W}_0\bigcup\mathcal{W}_1\bigcup\mathcal{W}_2\bigcup \mathcal{W}_3$.  In Fig. \ref{ReferenceConfiguration}, the agents represented by $\mathcal{W}_0$, $\mathcal{W}_1$, $\mathcal{W}_2$, and $\mathcal{W}_3$ are colored by cyan, red, green, and black, respectively.

\begin{figure}[ht]
    \includegraphics[width=\linewidth]{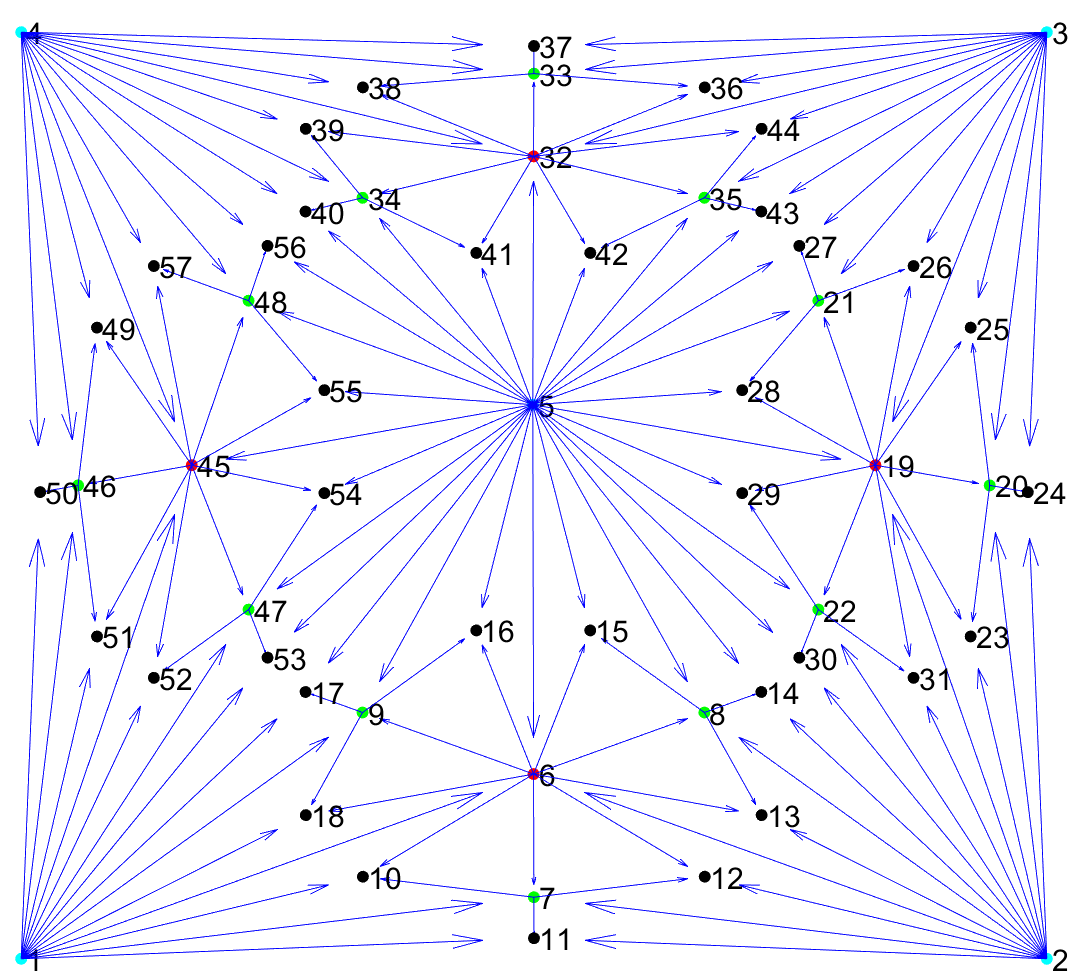}
    \centering
    \caption{Reference Configuration of the quadcopter team in a horizontal plane parallel to the $x-y$ plane. The inter-agent communication are directional and shown by blue arrays.}
    \label{ReferenceConfiguration}
\end{figure}
\begin{figure}[ht]
    \includegraphics[width=\linewidth]{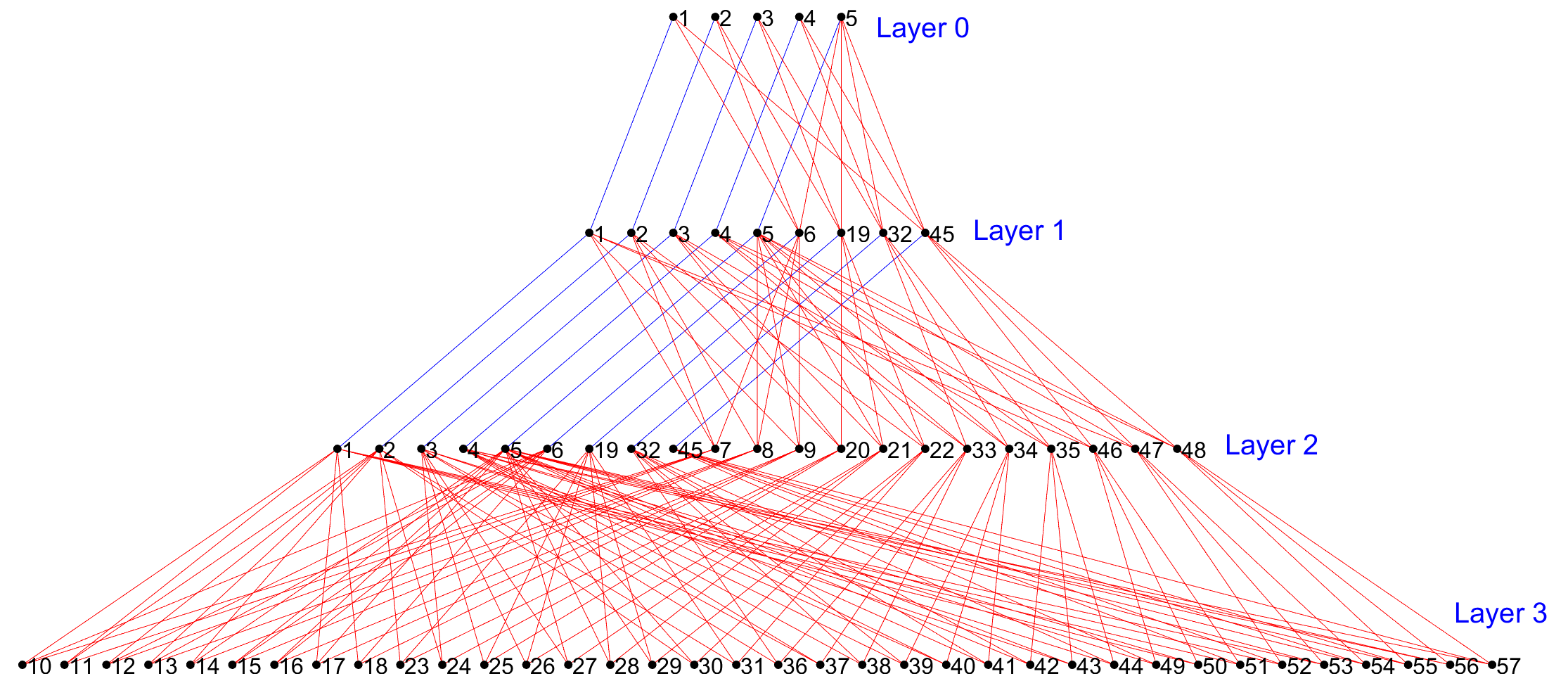}
    \centering
    \caption{The feed-forward network used for specifying inter-agent communication among the quadcopter team.}
    \label{NNResults}
\end{figure}

\begin{figure}[h]
\centering
 \subfigure[]{\includegraphics[width=0.98\linewidth]{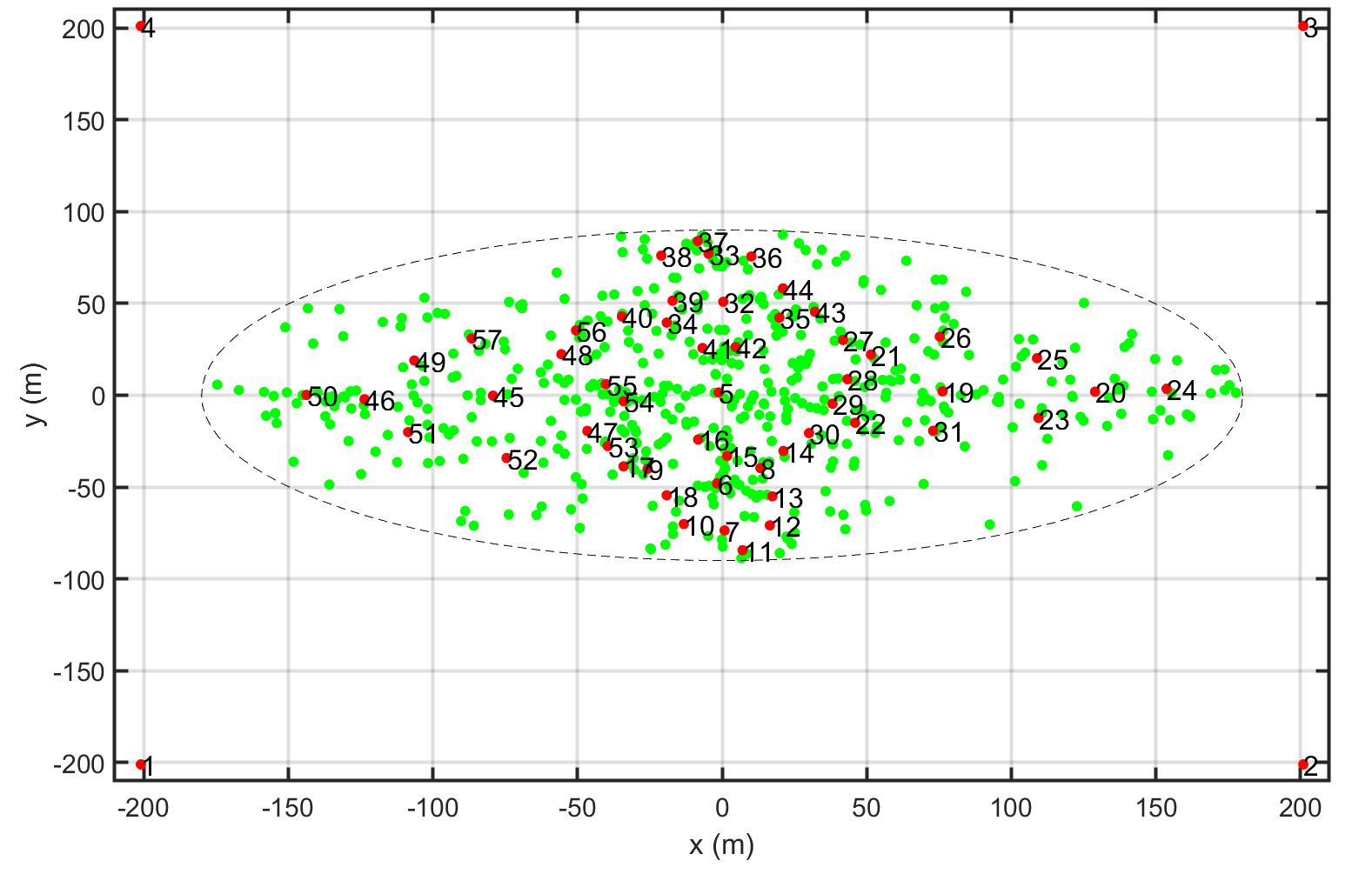}}
 \subfigure[]{\includegraphics[width=0.98\linewidth]{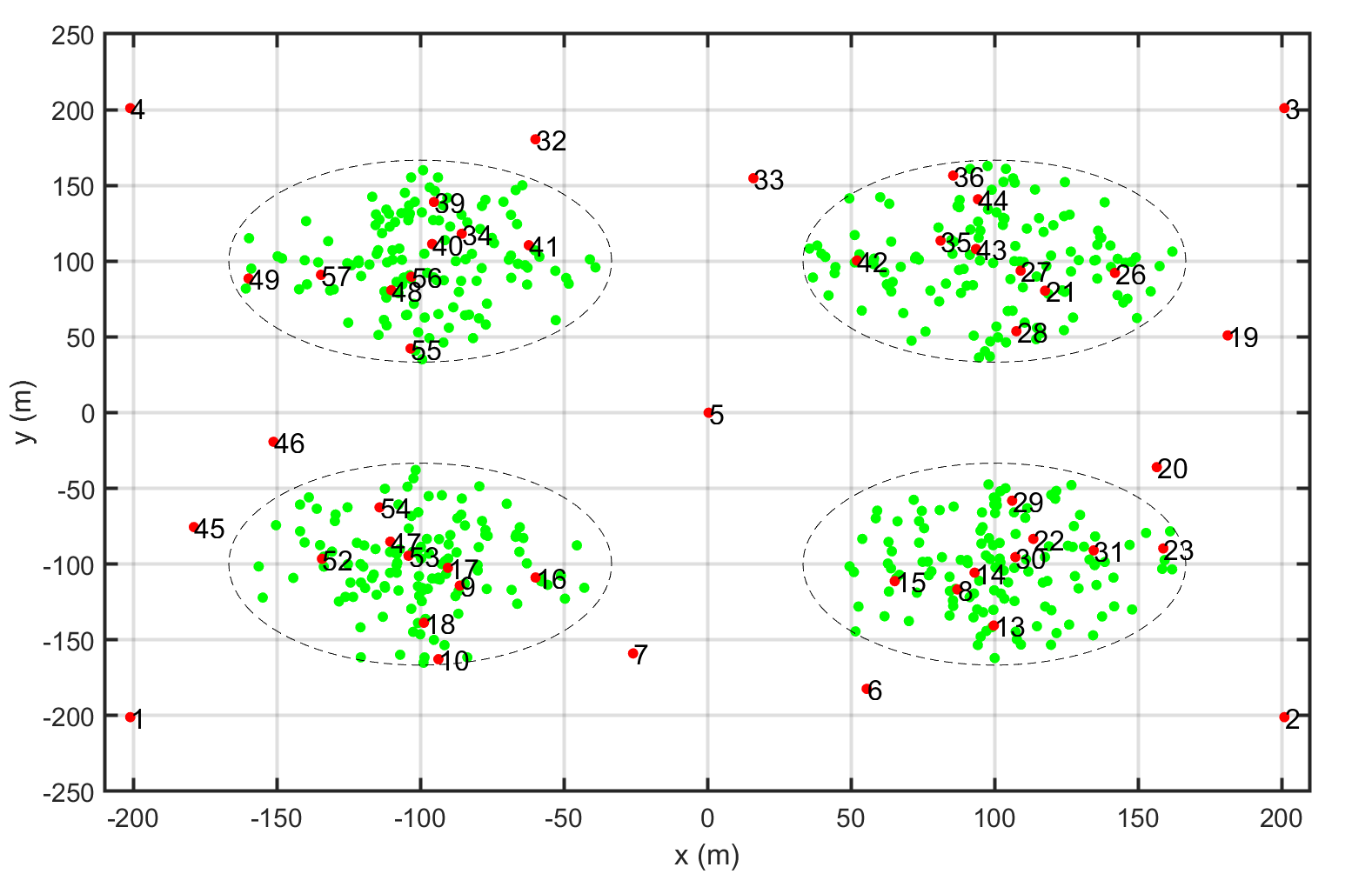}}
 \subfigure[]{\includegraphics[width=0.98\linewidth]{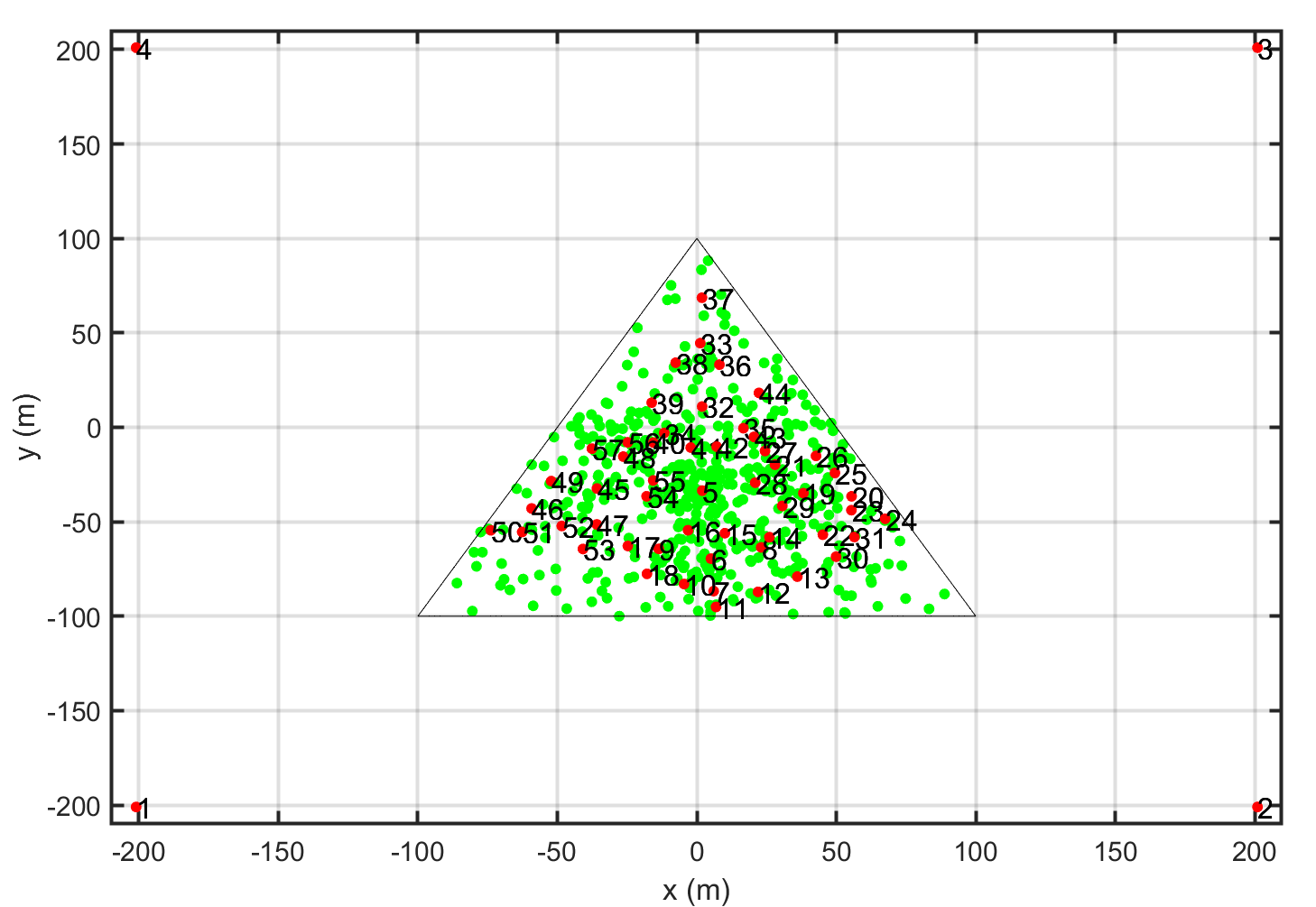}}
 \vspace{-.5cm}
\caption{Desired (a) elliptic, (b) multi-circular, and (c) triangular target distributions. The red spots are the abstract representation of the target sets by $52$ nodes.}
\label{coveragedata}
\end{figure}

\begin{figure}[h]
\centering
 \subfigure[]{\includegraphics[width=0.32\linewidth]{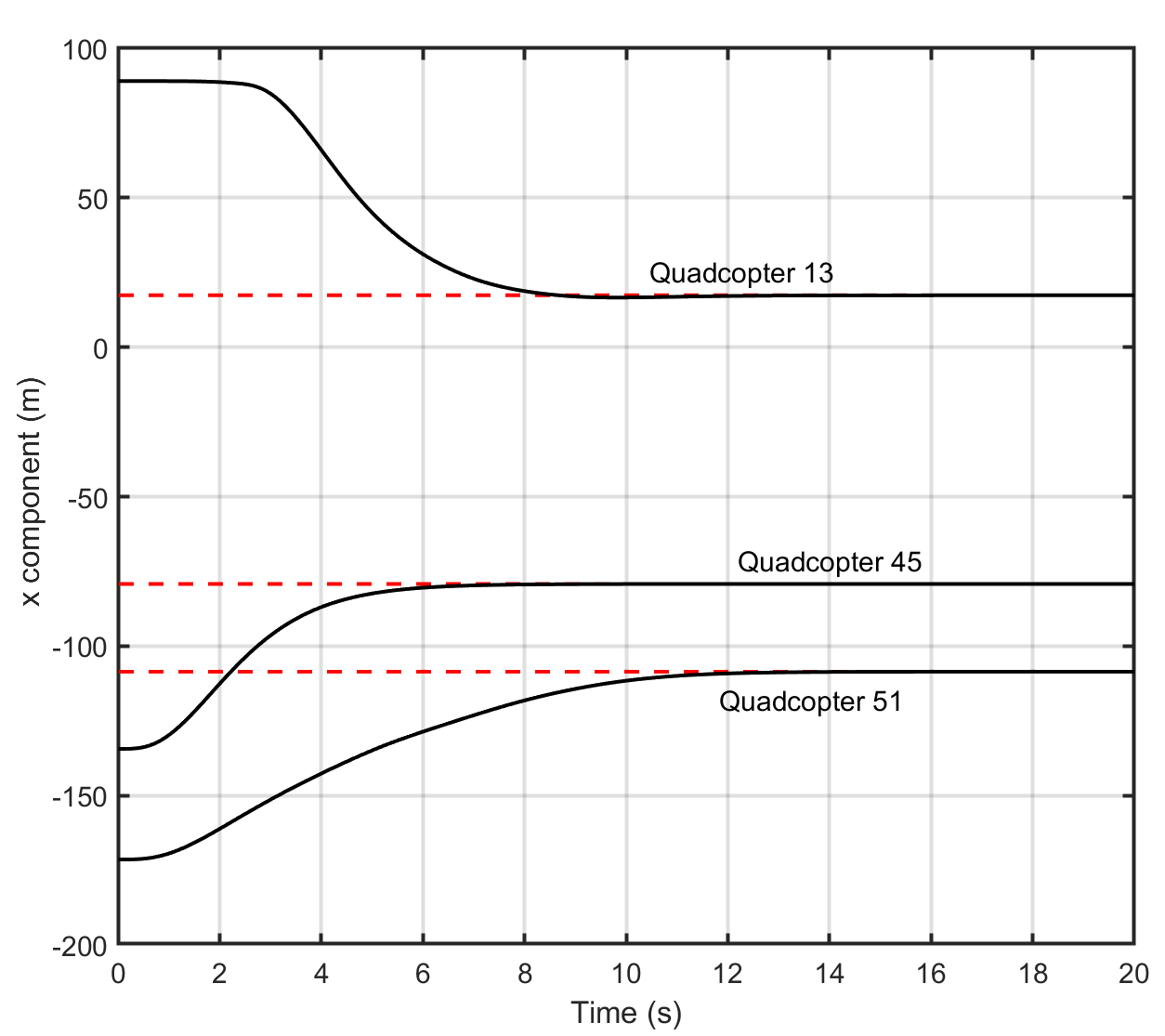}}
 \subfigure[]{\includegraphics[width=0.32\linewidth]{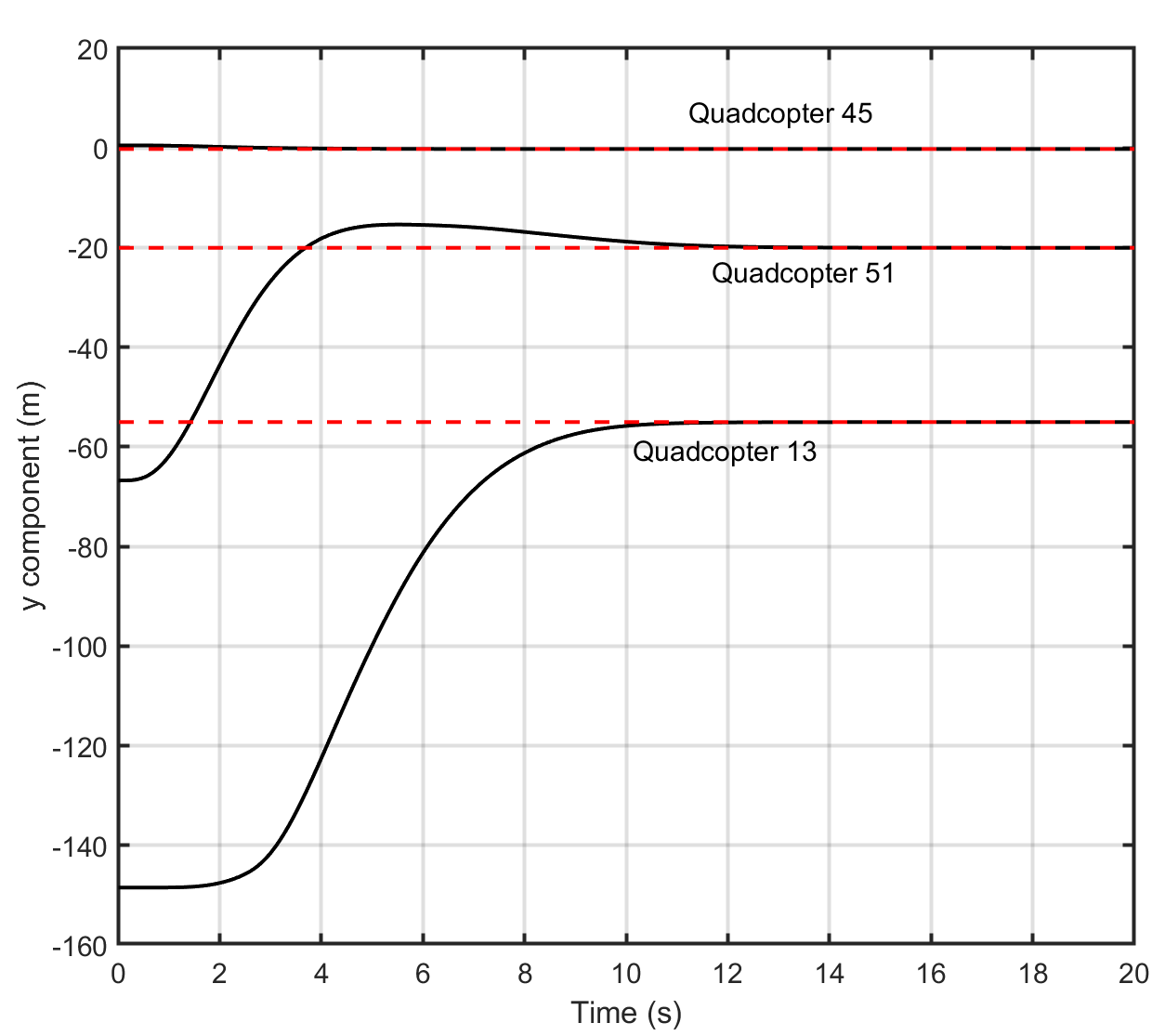}}
 \subfigure[]{\includegraphics[width=0.32\linewidth]{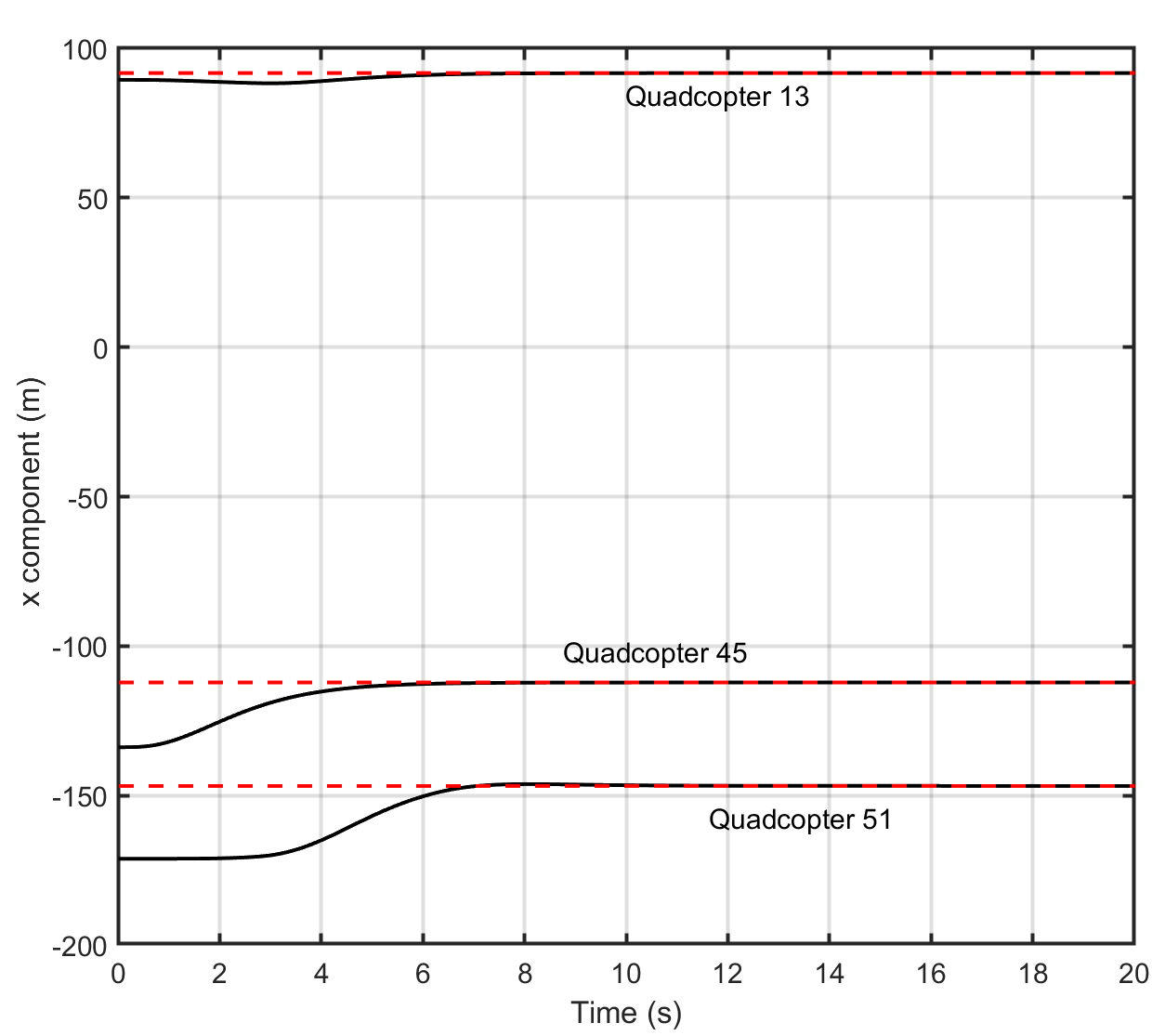}}
  \subfigure[]{\includegraphics[width=0.32\linewidth]{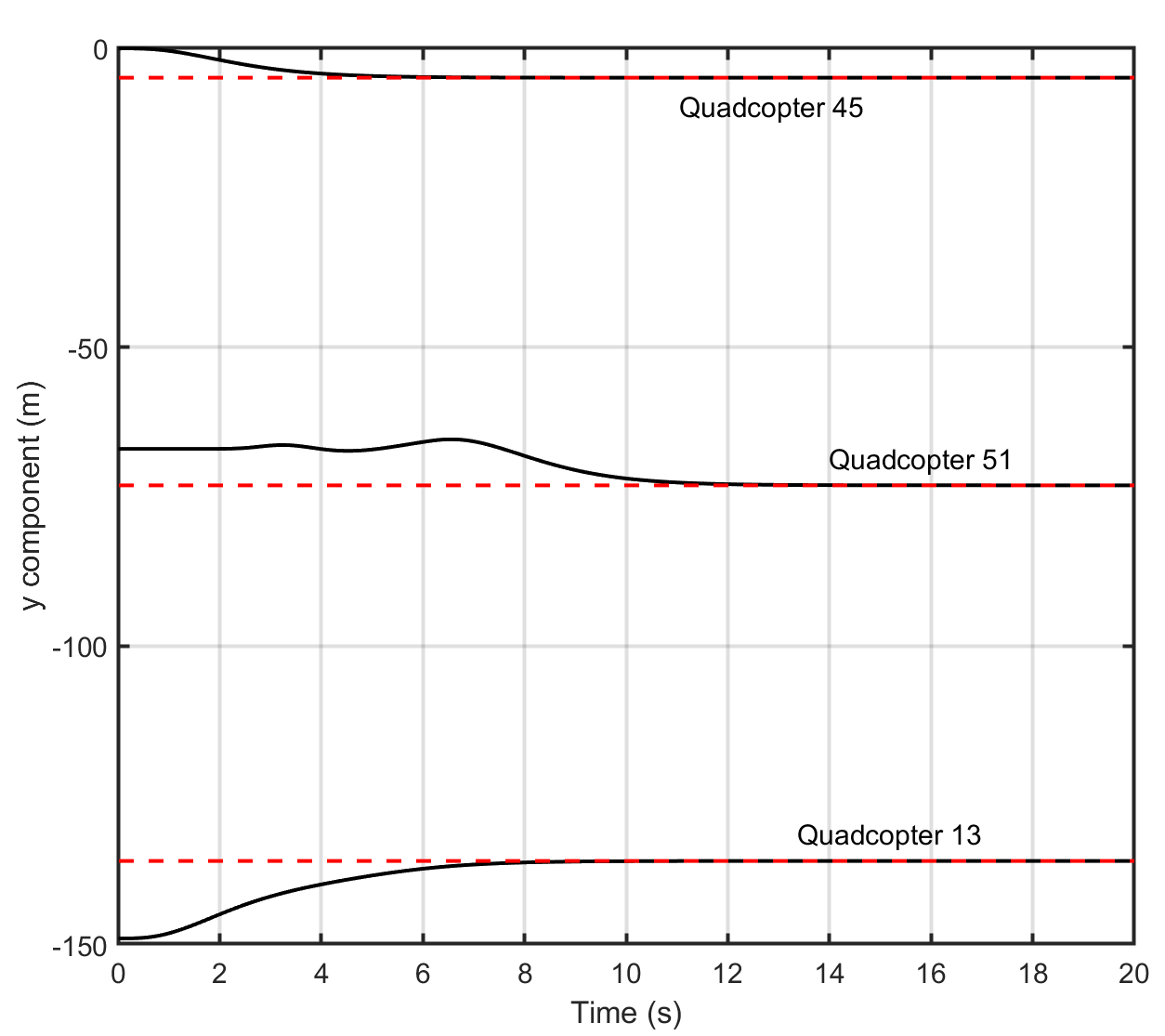}}
 \subfigure[]{\includegraphics[width=0.32\linewidth]{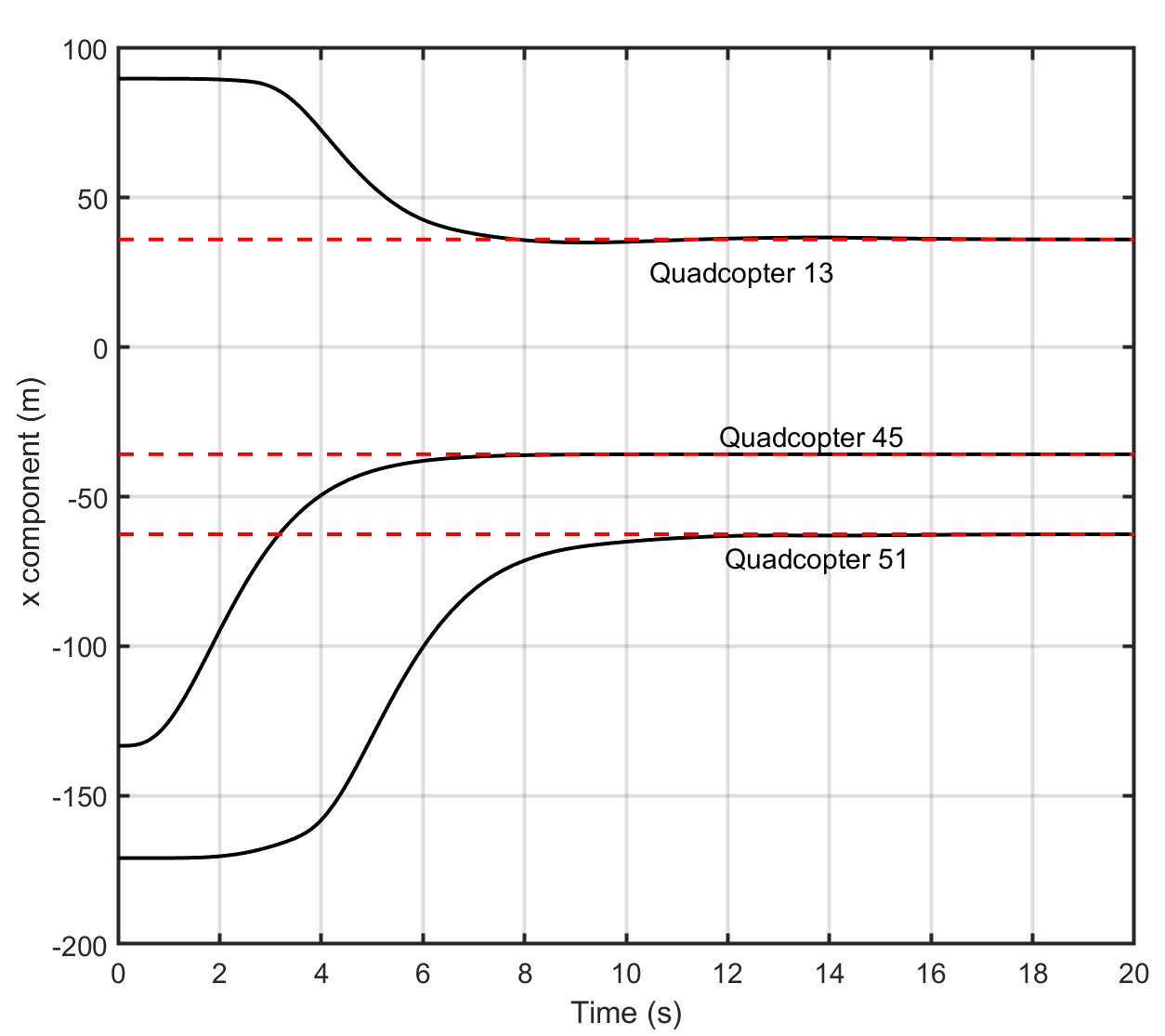}}
 \subfigure[]{\includegraphics[width=0.32\linewidth]{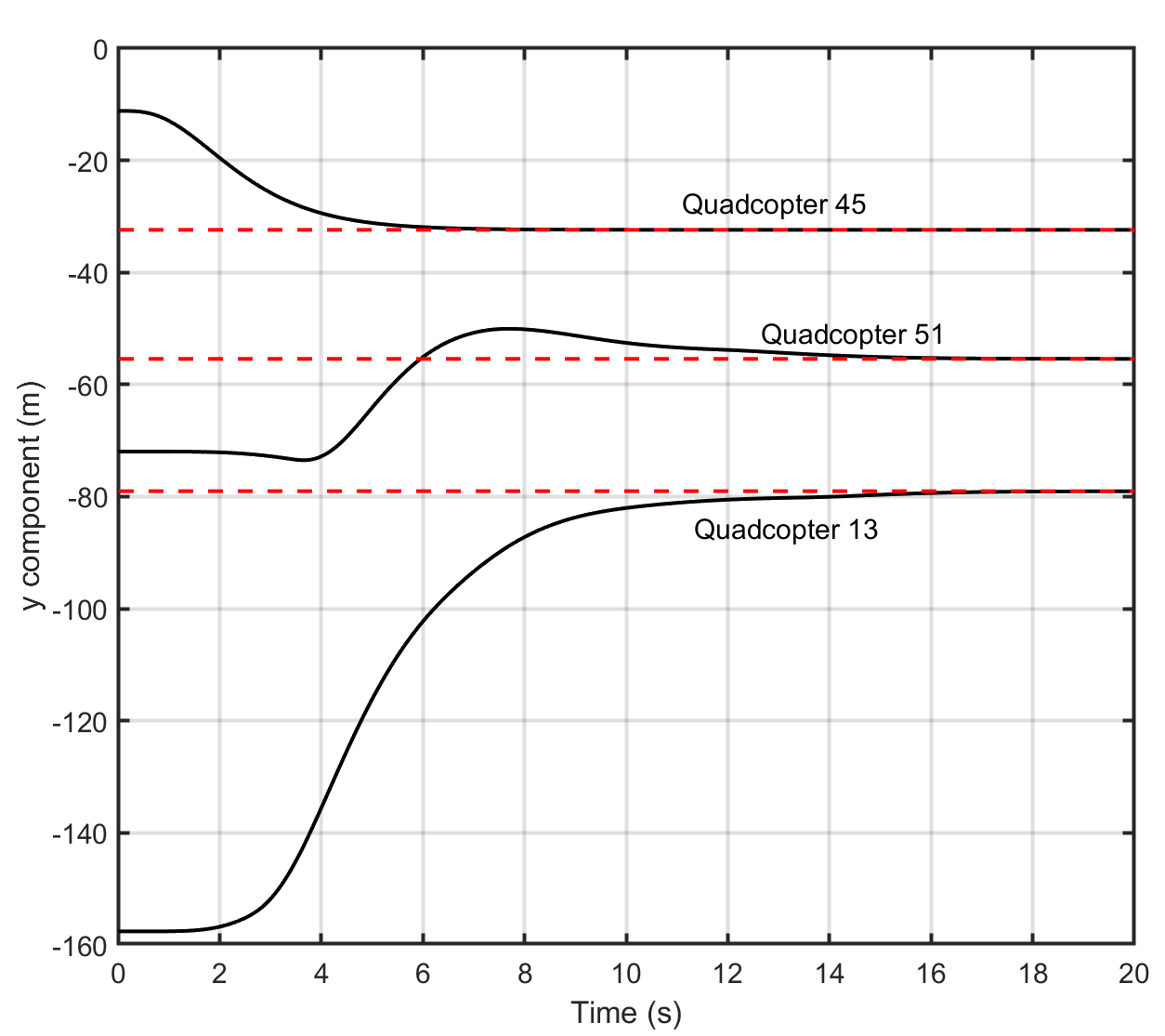}}
\caption{(a,c,e) $x$ components  of actual position of agents $13$, $45$, and $51$ over time interval $[0,20]s$ shown by solid black curves. $x$ components  of desired position of agents $13$, $45$, and $51$ over time interval $[0,20]s$ shown by dashed red curves. (b,d,f) $y$ components  of actual position of agents $13$, $45$, and $51$ over time interval $[0,20]s$ shown by solid black curves. $y$ components  of desired position of agents $13$, $45$, and $51$ over time interval $[0,20]s$ shown by dashed red curves.}
\label{coordinate}
\end{figure}

We apply the proposed coverage algorithm to cover elliptic, multi-circle, and triangular zones, each specified by the corresponding data set $\mathcal{D}$, where $\mathcal{D}$ defines $500$ data points shown by green spots in Figs. \ref{coveragedata} (a,b,c). As shown,  each target set is represented by $52$ points  positioned at $\mathbf{p}_6$ through $\mathbf{p}_{57}$, where they are obtained by using the approach presented in Section \ref{Assignment of Leaders' Positions}. These points are shown by red in Figs. \ref{coveragedata} (a,b,c). 

Figures \ref{coordinate} shows the components of actual and desired positions of quadcopters $13$, $45$, and $51$ are plotted versus time overt time interval $[0,20]s$, by solid black and dashed red, respectively. As seen the actual position of these three agents almost reach the designated desired positions at time $t=12s$. Figure \ref{comweghtconvergence} shows the time-varying communication weights of agent $41$ with its in-neighbors defined by $\mathcal{N}_{41}=\left\{34,5,32\right\}$. As shown, $w_{41,j}(t)$ converges to its desired value of $\varpi_{41,j}$ in about $12$ seconds for every $j\in \mathcal{N}_{41}$.

\begin{figure}[ht]
    \includegraphics[width=\linewidth]{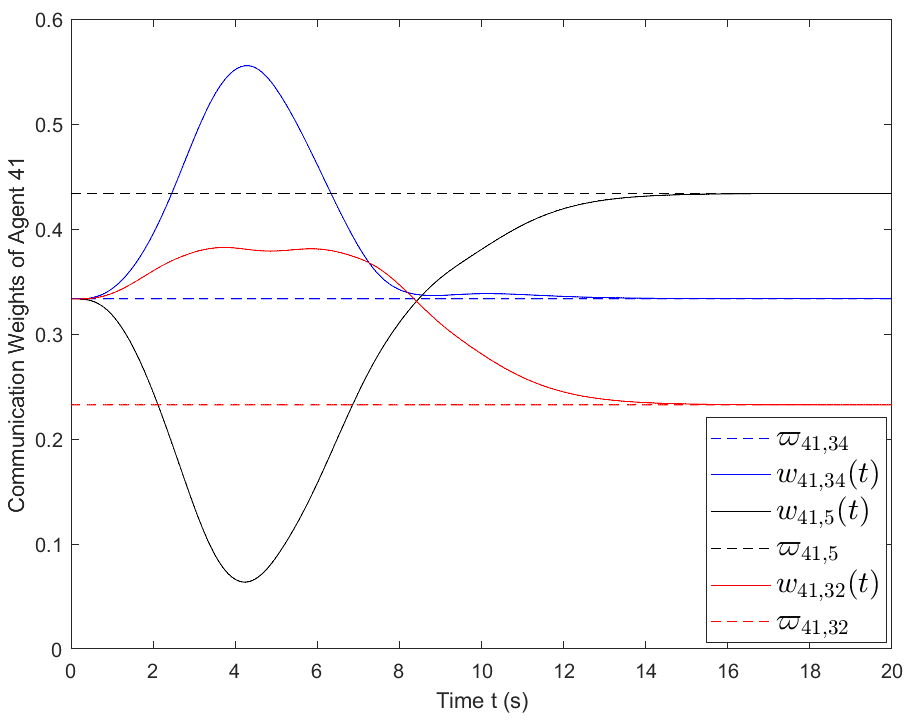}
    \centering
    \caption{Communication weights of agent $41$ with in-neighbor agents $34$, $5$, and $32$. The time varying communication weights $w_{41,34}(t)$, $w_{41,5}(t)$, and $w_{41,32}(t)$ converge to the desired values $\varpi_{41,34}$, $\varpi_{41,5}$, and $\varpi_{41,32}$ in about $12$ seconds.}
    \label{comweghtconvergence}
\end{figure}

\section{Conclusion}\label{Conclusion}
We proposed a novel neural-network-based approach for multi-agent coverage of a target with unknown distribution. We developed a forward approach to train the weights of the coverage neural network such that: (i) the target is represented by a finite number of points, (ii) the multi-agent system quickly and decentralizedly converge to the designated points representing the target distribution. For validation,  we performed a simulation of multi-agent coverage using a team of $57$ quadcopters, each of which is represented by at least one neuron of a the coverage neural network. The simulation results verified  fast and decentralized convergence of the proposed multi-agent coverage where each quadcopter reached its designated desired position in about $12$ seconds.

\bibliographystyle{IEEEtran}
\bibliography{reference}
\appendix
Let $x_i$, $y_i$, and $z_i$ denote position components of quadcopter $i\in 
\mathcal{V}$, and $p_i$, $m_i$ $\psi_i$, $\theta_i$, and $\psi_i$ denote the thrust force magnitude, mass, roll, pitch,  yaw angles of quadcopter $i\in \mathcal{V}$, and $g=9.81m/s^2$ be the gravity acceleration.   Then, we can use the model developed in \cite{rastgoftar2021safe, asslouj2022quadcopter}  and present the quadcopter dynamics by 
\begin{equation}
\label{generalnonlineardynamics}
\dot{\mathbf{x}}_i=\mathbf{f}\left(\mathbf{x}_i,\mathbf{u}_i\right)
    ,
\end{equation}
where $\mathbf{f}\left(\mathbf{x}_i,\mathbf{u}_i\right)=\mathbf{F}\left(\mathbf{x}_i\right)+\mathbf{G}\left(\mathbf{x}_i\right)\mathbf{u}_i$
\begin{equation}
\resizebox{0.99\hsize}{!}{%
$
    \mathbf{x}_i=\begin{bmatrix}
x_i&y_i&z_i&\dot{x}_i&\dot{y}_i&\dot{z}_i&\phi_i&\theta_i&\psi_i&\dot{\phi}_i&\dot{\theta}_i&\dot{\psi}_i&p_i&\dot{p}_i
\end{bmatrix}
^T,
$
}
\end{equation}
\begin{equation}
    \mathbf{u}_i=\begin{bmatrix}
u_{1,i}&u_{2,i}&u_{3,i}&u_{4,i}
\end{bmatrix}
^T,
\end{equation}
\begin{equation}
{\color{black}\mathbf{F}\left(\mathbf{x}_i\right)=\begin{bmatrix}\dot{x}_i\\
    \dot{y}_i\\
    \dot{z}_i\\
    {p_i\over m}\left(\sin{\phi_i}\sin{\psi_i} + \cos{\phi_i}\cos{\psi_i}\sin{\theta_i}\right)\\
    {p_i\over m}\left(\cos{\phi_i}\sin{\psi_i}\sin{\theta_i}- \sin{\phi_i}\cos{\psi_i}\right)\\
    {p_i\over m}\cos{\phi_i}\cos{\theta_i}-9.81\\
    \dot{\phi}_i\\
    \dot{\theta}_i\\
    \dot{\psi}_i\\
    0\\
    0\\
    0\\
    \dot{p}_i\\
    0\\
    \end{bmatrix}
    }
,
\end{equation}
\begin{equation}
    {\mathbf{G}}\left({\mathbf{x}}_i\right)=\begin{bmatrix}
    {\mathbf{g}}_1&{\mathbf{g}}_2&{\mathbf{g}}_3&{\mathbf{g}}_4
    \end{bmatrix}
    =
\begin{bmatrix}
\mathbf{0}_{9\times 1}&\mathbf{0}_{9\times 3}\\
\mathbf{0}_{3\times 1}&\mathbf{I}_3\\
0&\mathbf{0}_{1\times 3}\\
1&\mathbf{0}_{1\times 3}\\
\end{bmatrix}
,
\end{equation}
By defining transformation $\mathbf{x}_i\rightarrow\left(\mathbf{r}_i,\dot{\mathbf{r}}_i,\ddot{\mathbf{r}}_i,\dddot{\mathbf{r}}_i,\psi_i,\dot{\psi}_i\right)$, we can use the input-state feedback linearization approach presented in \cite{rastgoftar2021safe} and convert the the quadcopter dynamics to the following external dynamics:
\begin{subequations}
    \begin{equation}\label{quaddynamicsext}
        \ddddot{\mathbf{r}}_i=\mathbf{v}_i,
    \end{equation}
    \begin{equation}
        \ddot{\psi}_i=u_{\psi,i},
    \end{equation}
\end{subequations}
where $\mathbf{v}_i$ is related to the control input of quadcopter $i\in \mathcal{V}$, denoted by $\mathbf{u}_i$, by \cite{rastgoftar2022integration}
\begin{equation}
    \mathbf{v}_i=\mathbf{M}_{1,i}\mathbf{u}_i+\mathbf{M}_{2,i},
\end{equation}
with
\begin{subequations}
\begin{equation}
  \mathbf{M}_{1,i}=  \begin{bmatrix}
        L_{{\mathbf{g}}_{_{1}}}L_{{\mathbf{f}}}^{3}x_i& L_{{\mathbf{g}}_{_{2}}}L_{{\mathbf{f}}}^{3}x_i& L_{{\mathbf{g}}_{_{3}}}L_{{\mathbf{f}}}^{3}x_i& L_{{\mathbf{g}}_{_{4}}}L_{{\mathbf{f}}}^{3}x_i\\
        L_{{\mathbf{g}}_{_{1}}}L_{{\mathbf{f}}}^{3}y_i& L_{{\mathbf{g}}_{_{2}}}L_{{\mathbf{f}}}^{3}y_i& L_{{\mathbf{g}}_{_{3}}}L_{{\mathbf{f}}}^{3}y_i& L_{{\mathbf{g}}_{_{4}}}L_{{\mathbf{f}}}^{3}y_i\\
        L_{{\mathbf{g}}_{_{1}}}L_{{\mathbf{f}}}^{3}z_i& L_{{\mathbf{g}}_{_{2}}}L_{{\mathbf{f}}}^{3}z_i& L_{{\mathbf{g}}_{_{3}}}L_{{\mathbf{f}}}^{3}z_i& L_{{\mathbf{g}}_{_{4}}}L_{{\mathbf{f}}}^{3}z_i\\
         L_{{\mathbf{g}}_{_{1}}}L_{{\mathbf{f}}}\psi_i& L_{{\mathbf{g}}_{_{2}}}L_{{\mathbf{f}}}\psi_i& L_{{\mathbf{g}}_{_{3}}}L_{{\mathbf{f}}}\psi_i& L_{{\mathbf{g}}_{_{4}}}L_{{\mathbf{f}}}\psi_i\\
        \end{bmatrix}
        \in \mathbb{R}^{14\times 14}
        ,
\end{equation}
\begin{equation}
  \mathbf{M}_{2,i}=  \begin{bmatrix}
        L_{{\mathbf{f}}}^{4}x_i&
       L_{{\mathbf{f}}}^{4}y_i&
        L_{{\mathbf{f}}}^{4}z_i&
         L_{{\mathbf{f}}}^2\psi_i
        \end{bmatrix}
        ^T\in \mathbb{R}^{14\times 1}
        .
\end{equation}
\end{subequations}
In this paper, we assume that the desired yaw angle and its time derivative are both zero at any time $t$, and choose
\begin{equation}
    u_{\psi,i}=-k_5\dot{\psi}_i-k_6{\psi}_i
\end{equation}
Therefore, we can assume that $\psi_i(t)=0$ at any time $t$, as a result, the quadcopter $i\in \mathcal{V}$ can be modeled by Eq. \eqref{quaddynamicsext}.

\begin{IEEEbiography}[{\includegraphics[width=1in,height=1.25in,clip,keepaspectratio]{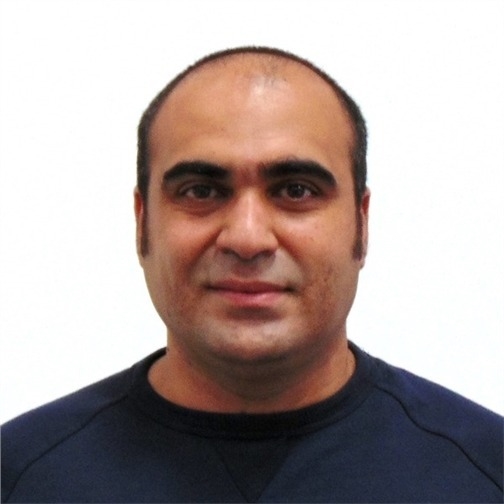}}]
{\textbf{Hossein Rastgoftar}} an Assistant Professor at the University of Arizona. Prior to this, he was an adjunct Assistant Professor at the University of Michigan from 2020 to 2021. He was also an Assistant Research Scientist (2017 to 2020) and a Postdoctoral Researcher (2015 to 2017) in the Aerospace Engineering Department at the University of Michigan Ann Arbor. He received the B.Sc. degree in mechanical engineering-thermo-fluids from Shiraz University, Shiraz, Iran, the M.S. degrees in mechanical systems and solid mechanics from Shiraz University and the University of Central Florida, Orlando, FL, USA, and the Ph.D. degree in mechanical engineering from Drexel University, Philadelphia, in 2015. His current research interests include dynamics and control, multiagent systems, cyber-physical systems, and optimization and Markov decision processes.
\end{IEEEbiography}
\end{document}